\documentclass[a4paper,11pt]{article}

\usepackage[english]{babel}
\usepackage[ansinew]{inputenc}
\usepackage{fontenc}
\usepackage{amsfonts}
\usepackage{amsmath}
\usepackage{amsthm}
\usepackage{enumerate}
\usepackage{textcomp}
\usepackage{geometry}
\usepackage{mathrsfs}
\usepackage{natbib}
\usepackage{dsfont} 
\usepackage{graphicx}
\usepackage{subcaption}
\usepackage{multirow}
\usepackage[ruled]{algorithm2e}

\usepackage[colorlinks=true,linkcolor=red,citecolor=blue]{hyperref}

\numberwithin{equation}{section}

\newcommand{\ind}[1]{\ensuremath{\mathds{1}_{#1}}}
\newcommand{\C}{\mathcal{C}}

\newtheorem{lemma}{Lemma}[section]
\newtheorem{proposition}{Proposition}[section]

\theoremstyle{definition}
\newtheorem{remark}{Remark}[section]

\begin{document}

\author{Diaa Al Mohamad\footnote{Corresponding author. E-mail: diaa.almohamad@gmail.com} \;\; \; \;\; Jelle J. Goeman \;\;\;\;\; Erik W. van Zwet
 \\ \normalsize{Leiden University Medical Center, The Netherlands} \\
Eric A. Cator \\
\normalsize{Radboud University, The Netherlands}
\\
Last update 
}

\title{Adaptive Critical Value for Constrained Likelihood Ratio Testing}

\maketitle

\begin{abstract} 
We present a new way of testing ordered hypotheses against all alternatives which overpowers the classical approach both in simplicity and statistical power. Our new method tests the constrained likelihood ratio statistic against the quantile of one and only one chi-squared random variable with a data-dependent degrees of freedom instead of a mixture of chi-squares. Our new test is proved to have a valid finite-sample significance level $\alpha$ and provides more power especially for sparse alternatives (those with a few or moderate number of null constraints violations) in comparison to the classical approach. Our method is also easier to use than the classical approach which requires to calculate or simulate a set of complicated weights. Two special cases are considered with more details, namely the case of testing orthants $\mu_1<0, \cdots, \mu_n<0$ and the isotonic case of testing $\mu_1<\mu_2<\mu_3$ against all alternatives. Contours of the difference in power are shown for these examples showing the interest of our new approach. \\ 

\textbf{Keywords:} Polyhedral cones, Ordered hypothesis, Conditional likelihood ratio, Isotonic regression, PAVA, Orthants, Power. 
\end{abstract}
 
\tableofcontents
\section*{Introduction}
The problem of testing ordered hypotheses or constrained likelihood ratio tesing debuted mainly in the 50's of the $20^{th}$ century with the works of \cite{Brunk}, \cite{vanEeden} and \cite{BartholomewPAVA} (see for example \citet{ReviewRobertsonBook} for a brief history). The developments achieved by several authors in the next two decades were put together in the book \cite{BarlowBook} where they test the hypothesis $\mu_1=\cdots=\mu_p$ against some order restriction on the real vector $\mu=(\mu_1,\cdots,\mu_p)$. Later developments came to light afterwards expanding the testing problem to more complicated forms of nulls and alternatives especially testing if a vector $\mu$ is in some linear space against the alternative that it is inside some cone. The book of \cite{RobertsonBook} was published to cover for these developments. We call this kind of problems (following the notation in \citet{SilvapulleBook}) "type A" problems. Another type of testing problems had already surfaced in some papers such as \cite{Robertson78} and was then generalized to a multivariate context in papers such as \cite{Shapiro}. In this kind of problems that we call "type B" problems, the null hypothesis defines an order over the the vector of interest $\mu$ of the form $A\mu\geq 0$ for some matrix $A$, and test it against all alternatives. The most recent findings and developments concerning type B problems were gathered and published in the book of \cite{SilvapulleBook}. The book covers as well a more general approach for these testing problems and is the main reference in our paper. \\
Applications of constrained likelihood ratio testing could be found in various fields. For example in linear regression, it is sometimes relevant to assume that the parameters related to some ordinal covariate follow a specific order. According to \cite{IcInferPackage}, a great attention to constrained LRT could be found in econometrics where a software called GAUSS (Aptech Systems Inc. 2009) is developed for this purpose. Several interesting examples from different fields of study could be found in \cite[Chap. 1]{SilvapulleBook}.\\
Let $Y\sim\mathcal{N}_p(\mu,V)$. In the literature, the constrained LR for testing $A\mu\geq 0$ against all alternatives (type B) has a least favorable distribution when $\mu=0$. The distribution is a mixture of chi-squares with weights $w_1,\cdots,w_K$ (for some $K$) which depend on both the matrix $A$ and the covariance structure $V$. This makes these weights very difficult to calculate for $p\geq 5$. Very few special cases for the matrix $A$ and when there is no dependence structure in the data, that is $V=I$, were solved in the literature and explicit formulas for the weights could be found (\citet{Miles59}, \citet[Section 3.5]{SilvapulleBook}, \citet[Examples 15,16]{Drton}, \citet{HoshinoIsotonTree}). For the general case, the R package \texttt{ic.infer} (\citet{IcInferPackage}) provides functionalities for this calculation. \cite{IcInferPackage} provides insight on how difficult and computationally extensive this calculation becomes when the dimension $p$ increases.\\
There have appeared in the literature some papers proposing a different approach for constrained LRT where instead of testing against a quantile of a mixture, they test against the quantile of one chi-square with data-dependent degrees of freedom (\citet{BartholomewConditional}, \citet{Susko}, \citet{ChenConditional}, \cite{WollanDykstra},\citet{IversonHarp},\citet{Rueda}). This approach avoids the calculation of the weights so that it is computationally very easy, but this comes at the cost of a reduction in power for type A problems (\citet{BartholomewConditional}, \citet{Susko}, \citet{ChenConditional}). For type B problems, this approach was proposed in the isotonic case of testing $\mu_1\leq\cdots\leq\mu_p$ against all alternatives independently by \cite{WollanDykstra} and \citet{IversonHarp} (see \citet{Rueda} for an application). Authors of these papers conjectured that the test has a valid $\alpha$ level and provided only asymptotic evidence (that is when the common variance go to zero) that it the significance level is valid and verified it through simulations. They argue that not only that the new approach is computationally more feasible (because no weights need to be calculated), but that it is also more powerful than the classical approach of testing against a quantile of a mixture of chi-squares in some regions of the space of the parameters.\\

We consider in this paper the type B testing problem that is to test $A\mu\geq 0$ for some matrix $A$ against all alternatives using the LRT. We adopt the approach proposed by \cite{WollanDykstra} and \citet{IversonHarp} for testing the LR statistic against the quantile of one chi-square with data-dependent degrees of freedom. We prove that the corresponding test has indeed a valid $\alpha$ level and not only asymptotically, and thus prove the conjecture given by \cite{WollanDykstra} and \citet{IversonHarp} for the isotonic case. We argue as well that the power of the new approach is indeed greater than the classical method of testing with mixtures and show more details and insights for the case of testing $\mu\leq 0$ and for the isotonic case. For example, for testing $\mu\leq 0$, it is shown that as long as the true vector $\mu$ does not violate more than half the number of constraints that is $p/2$, then our approach is more powerful than the classical method of mixtures.\\
The case when $V=\sigma^2\Sigma$ with $\Sigma$ is known and $\sigma^2$ is unknown is a straightforward generalization of our approach and is thus considered in this paper.\\

The paper is organized as follows. In Section \ref{sec:context}, we introduce briefly the testing problem considered in this paper, namely type B, with the used notations. In Sections \ref{sec:typeB2VarQuad}-\ref{sec:typeBGeneral}, we provide our new solution for testing type B problems first in special cases moving towards the general case. A comparison of the difference in power between the classical approach in the literature and our new approach is also illustrated through these sections in special cases. More details are given for the special case of orthants hypotheses and isotonic ones in Section \ref{sec:Examples}. The case of known covariance matrix is discussed in Section \ref{sec:UnKnownVar}. Finally, in Section \ref{sec:PowerCompare} we provide some arguments to when our new approach overpowers the classical one. Appendices \ref{append1}, \ref{append2}, \ref{append3} and \ref{append4} include the type A testing problem already consided in the literature and some alternative proofs for special cases of type B problems.


\section{The problem of inequality constrained testing: Context and notations}\label{sec:context}
We say that $\mathcal{C}$ is a polyhedral (cone) if there exist $a_1,\cdots,a_k\in\mathbb{R}^p$ such that 
\begin{align*}
\C=&\{\mu\in\mathbb{R}^p: a_1^T\mu\geq 0,\cdots,a_k^T\mu\geq 0\} \\
 = & \{\mu\in\mathbb{R}^p: A\mu\geq 0\}.
\end{align*}
Note that a polyhedral is a closed convex cone. For $y\in\mathbb{R}^p$, we denote $\|y\|^2=y^ty$ the usual norm and $\|y\|_V^2=y^tV^{-1}y$ for a symmetric definite positive matrix $V$. The projection of $y$ on $\C$ is denoted $\Pi(y|\C)$ ($\Pi_V(y|\C)$, resp.) with respect to the usual norm (the norm $\|.\|_V$).
A Type B testing problem tests the null $H_1:\mu\in\mathcal{C}$ against $H_2:\mu\in\mathbb{R}^p$. This includes for example the isotonic hypothesis $H_1:\mu_1<\mu_2<\cdots<\mu_p$. Here $\C=\{\mu\in\mathbb{R}^p: \mu_{i+1}-\mu_i\geq 0, \forall i=1,\cdots,p-1\}$. \\
For notational simplicity, we denote $q_i$ the $(1-\alpha)$quantile of the chi-square distribution with $i$ degrees of freedom. For other orders, we use $q_{\chi^2(i)}(1-\bar{\alpha})$ for some $\bar{\alpha}\neq\alpha$.\\
Let $Y\sim\mathcal{N}(\mu,V)$. The LR statistic takes the form
\begin{equation}
LR(type B) = \min_{\mu\in\mathcal{C}} (Y-\mu)^TV^{-1}(Y-\mu)
\label{eqn:LRTypeB}
\end{equation}
and the least favorable null distribution is defined by
\begin{equation}
\mathbb{P}_{\mu=0}(LR(type B)\leq c) = \sum_{i=0}^p{w_{p-i}(p,V,C)\mathbb{P}(\chi^2(i)\leq c)}.
\label{eqn:ChiBarTypeB}
\end{equation}
%

\section{Adaptive quantile for type B testing}
We show in the following examples how the adaptive critical value is used and compare it with the classical way of dealing with constrained LRT. We show in each of the examples how we could prove that the significance level of the test with an adaptive critical value is exactly $\alpha$ before moving to the general case. 
\subsection{A one dimensional case}\label{subsec:OneVarTypeB}
We test $H_1:\mu_1\leq 0$ against all alternatives. The LR is given by
\[ LR(Y) = 
 \left\{\begin{array}{rcl}
0 & \text{if} & Y_1\leq 0 \\
q_1 & \text{if} & Y_1\geq 0
\end{array}\right.
\]
The classical approach tests the LR against the quantile $q$ such that
\begin{equation}
\mathbb{P}_{H_1}((Y_1-\mu_1)^2>q) = \frac{\alpha}{2}
\label{eqn:OneVarCrit}
\end{equation}
We propose to define the critical random value
\[ q(Y,\alpha) = 
 \left\{\begin{array}{rcl}
0 & \text{if} & Y_1\leq 0 \\
q_1 & \text{if} & Y_1\geq 0
\end{array}\right.
\]
We have
\[
\mathbb{P}_{H_1}\left(LR>q(Y,\alpha)\right) = \mathbb{P}_{H_1}\left(Y_1^2>q_1|Y_1\geq 0\right)\mathbb{P}_{H_1}\left(Y_1\geq 0\right) \]
 In the previous display, since the event $\{Y_1\geq 0\}$ implies the event $\{Y_1-\mu_1\geq 0\}$, then 
\[\mathbb{P}_{H_1}\left(Y_1\geq 0\right)\leq \mathbb{P}_{H_1}\left(Y_1-\mu_1\geq 0\right) =\frac{1}{2}\]
Moreover, using the fact that the Gaussian model has an increasing likelihood ratio, we could show that
\[\mathbb{P}_{H_1}\left(Y_1^2>q_1|Y_1\geq 0\right)\leq \mathbb{P}_{H_1}\left((Y_1-\mu_1)^2>q_1|Y_1-\mu_1>0\right)\]
Finally,
\[\mathbb{P}_{H_1}\left(LR>q(Y,\alpha)\right) \leq \frac{\alpha}{2}\]
In comparison to the classical approach (\ref{eqn:OneVarCrit}), it is clear that they both have the same statistical power under the alternative since they both test against the quantile of $\chi^2(1)$ of order $\alpha/2$.
\subsection{The two dimensional case}\label{sec:typeB2VarQuad}
Let $Y\sim(\mu,I_2)$ be a bivariate Gaussian random variable. Consider the testing problem
\[H_1: \mu_1\leq 0, \mu_2\leq 0, \text{ against } H_2: \mu\in \mathbb{R}^p.\]
We are testing if the vector $\mu$ is in the negative quadrant of the plane. The likelihood ratio for this test is given by
\[LR = \left\{\begin{array}{rcl}
0 & \text{if} & Y_1\leq 0, Y_2\leq 0 \\
Y_1^2 & \text{if} & Y_1\geq 0, Y_2\leq 0 \\
Y_2^2 & \text{if} & Y_1\leq 0, Y_2\geq 0 \\
Y_1^2+Y_2^2 & \text{if} & Y_1\geq 0, Y_2\geq 0.
\end{array}\right.\]
In the literature, testing $H_1$ against $H_2$ at level $\alpha$ using the LRT is done by looking for $c>0$ such that
\[\frac{1}{2}\mathbb{P}(\chi^2(1)>c) + \frac{1}{4}\mathbb{P}(\chi^2(2)>c) \leq \alpha.\]
The least favorable distribution of the LR is a mixture of the chi-squares $\chi^2(2), \chi^2(1)$ and $\chi^2(0)$ with respective weights $\frac{1}{4}, \frac{1}{2}, \frac{1}{4}$. We propose in this paper to do the test differently. Since the LR is conditionally distributed as a $\chi^2(1)$ when $Y$ has one negative coordinate and one positive one, then we will only test the LR against the quantile of the $\chi^2(1)$. When $Y$ is in the positive quadrant, the LR is conditionally distributed as a $\chi^2(2)$, and we test the LR against the quantile of the $\chi^2(2)$. In other words, we define the random quantile
\[ q(Y,\alpha) = 
 \left\{\begin{array}{rcl}
0 & \text{if} & Y_1\leq 0, Y_2\leq 0 \\
q_1 & \text{if} & Y_1\geq 0, Y_2\leq 0 \text{ or } Y_1\leq 0, Y_2\geq 0\\
q_2 & \text{if} & Y_1\geq 0, Y_2\geq 0.
\end{array}\right.
\]
We test $H_1$ against $H_2$ using the rejection region 
\[\left\{LR>q(Y,\alpha)\right\}.\]
Our claim is that 
\[\mathbb{P}_{H_1}\left(LR>q(Y,\alpha)\right)\leq \alpha.\]
We can decompose the probability of rejection according to the position of the vector $Y$ in the plane. Denote $Q_1,Q_3$ the nonnegative and non-positive quadrants respectively, and $Q_2, Q_4$ the quadrants $\{y\in\mathbb{R}^2, y_1<0, y_2>0\}$ and $\{y\in\mathbb{R}^2, y_1>0, y_2<0\}$. We have
\begin{multline}
\mathbb{P}_{H_1}\left(LR>q(Y,\alpha)\right) = \mathbb{P}_{H_1}\left(Y_1^2+Y_2^2>q_2| Y\in Q_1\right)\mathbb{P}_{H_1}\left(Y\in Q_1\right) + \\ 
\mathbb{P}_{H_1}\left(Y_1^2>q_1| Y\in Q_4\right)\mathbb{P}_{H_1}\left(Y\in Q_4\right)+ 
\mathbb{P}_{H_1}\left(Y_2^2>q_1| Y\in Q_2\right)\mathbb{P}_{H_1}\left(Y\in Q_2\right)
\label{eqn:All4Quartiles}
\end{multline}
Since $Y_1$ and $Y_2$ are assumed independent, we may write
\begin{multline*}
\mathbb{P}_{H_1}\left(LR>q(Y,\alpha)\right) = \mathbb{P}_{H_1}\left(Y_1^2+Y_2^2>q_2| Y\in Q_1\right)\mathbb{P}_{H_1}\left(Y\in Q_1\right) + \\ 
\mathbb{P}_{H_1}\left(Y_1^2>q_1| Y_1\geq 0\right)\mathbb{P}_{H_1}\left(Y\in Q_4\right)+ 
\mathbb{P}_{H_1}\left(Y_2^2>q_1| Y_2\geq 0\right)\mathbb{P}_{H_1}\left(Y\in Q_2\right)
\end{multline*}
We state that (see Lemma \ref{Lemm:ConditionalIncrease}) 
\begin{align*}
\mathbb{P}_{H_1}\left(Y_1^2>q_1| Y_1\geq 0\right) \leq & \mathbb{P}_{H_1}\left((Y_1-\mu_1)^2>q_1| Y_1-\mu_1\geq 0\right) \\
\mathbb{P}_{H_1}\left(Y_2^2>q_1| Y_2\geq 0\right) \leq & \mathbb{P}_{H_1}\left((Y_2-\mu_2)^2>q_1| Y_2-\mu_2\geq 0\right)\\
\mathbb{P}_{H_1}\left(Y_1^2+Y_2^2>q_2| Y\in Q_1\right) \leq & \mathbb{P}_{H_1}\left((Y_1-\mu_1)^2+(Y_2-\mu_2)^2>q_2| Y-\mu\in Q_1\right)
\end{align*}
Now, since the norm of a standard Gaussian random variable is independent from its direction (\citet[Lemma 3.13.1]{SilvapulleBook}), then each of the previous conditional probabilities is less than $\alpha$. Thus
\begin{align*}
\mathbb{P}_{H_1}(LR>q(Y,\alpha)) \leq & \alpha\left(\mathbb{P}_{H_1}\left(Y\in Q_1\right)+\mathbb{P}_{H_1}\left(Y\in Q_2\right)+\mathbb{P}_{H_1}\left(Y\in Q_4\right)\right) \\
\leq & (1-\mathbb{P}_{H_1}(Y\in Q_3))\alpha \\
\leq & \frac{3}{4}\alpha \\
< & \alpha.
\end{align*}
It is interesting to note that $q(Y,\alpha)$ could be redefined so that we exploit the remaining $\frac{1}{4}\alpha$ in order to make the test even more powerful. Therefore, we redefine it as follows
\[ \bar{q}(Y,\alpha) = q\left(Y,\frac{4}{3}\alpha\right) = 
 \left\{\begin{array}{rcl}
0 & \text{if} & Y_1\leq 0, Y_2\leq 0 \\
q_{\chi^2(1)}(1-\frac{4}{3}\alpha) & \text{if} & Y_1\geq 0, Y_2\leq 0 \text{ or } Y_1\leq 0, Y_2\geq 0\\
q_{\chi^2(2)}(1-\frac{4}{3}\alpha) & \text{if} & Y_1\geq 0, Y_2\geq 0.
\end{array}\right.
\]
For example for $\alpha=0.05$, we test against chi-squared quantiles at level $0.93$ instead of the usual level that is $0.95$. Figure (\ref{fig:PowerDiffTypeB}) shows the difference in power between our proposed method denoted as "Selective" and the method in the literature denoted as "Mixture".\\
Our new approach seems to be more powerful in the two quadrants $Q_2$ and $Q_4$ along the borders of the null hypothesis, whereas the classical approach is more powerful only in the positive quadrant $Q_1$ inside the circle centered at approximately $(2.5,2.5)$ with radius $2.5$. The contours show that in the regions where the adaptive critical value is used, the gain reaches $8\%$ whereas in the regions where the mixture quantile is used, the gain reaches only $7\%$. 
\begin{figure}[ht]
\centering
\includegraphics[scale = 0.6,clip,trim=2cm 2cm 2cm 1.6cm]{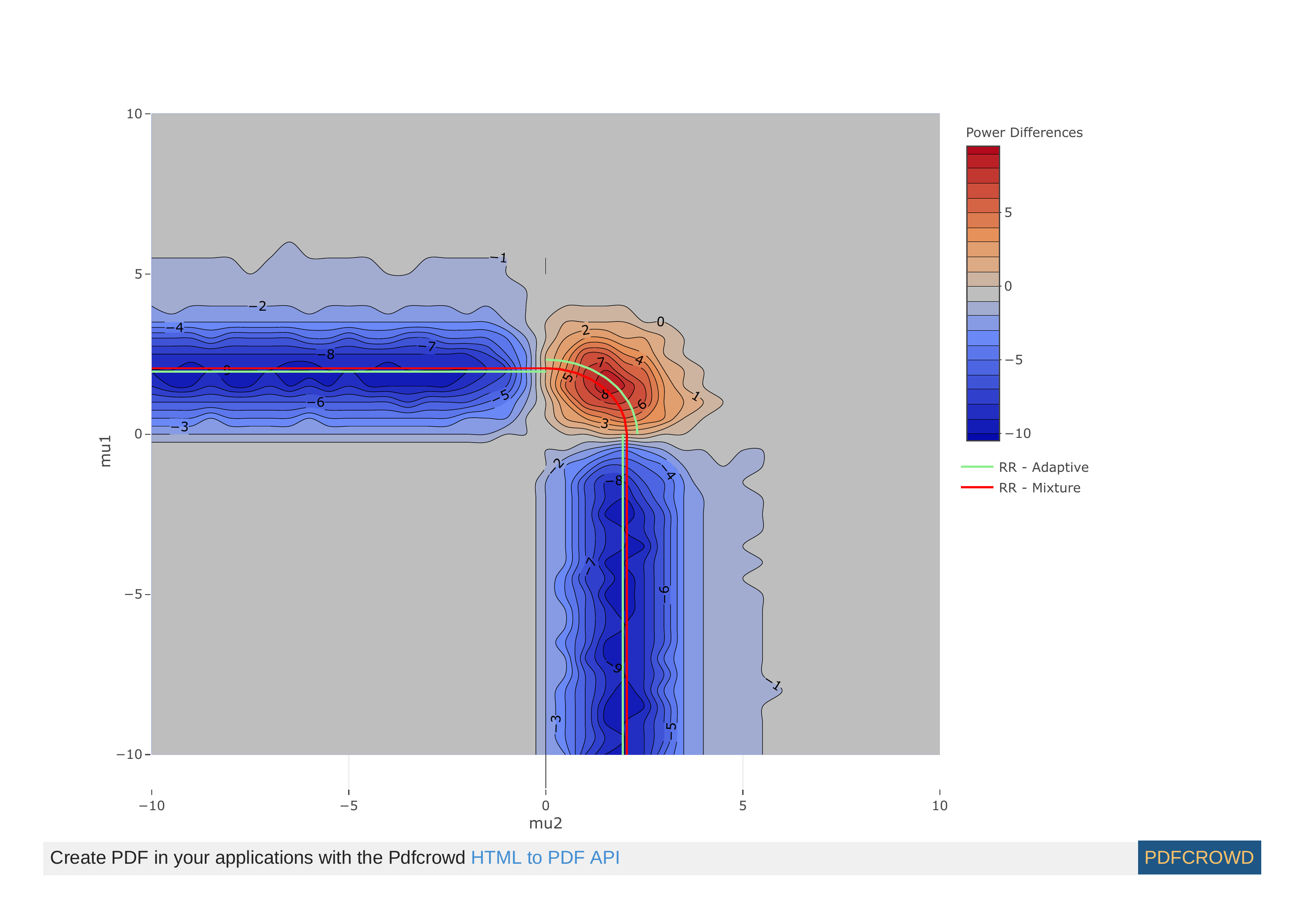}
\caption{Contours for the difference in power $100\times$(Classical$-$New). The red line represents the rejection line for the classical approach whereas the green one represents the rejection line for our approach. Figure produced using package \texttt{Plotly}}
\label{fig:PowerDiffTypeB}
\end{figure}

\subsection{The isotonic hypothesis with three means}
We test the null hypothesis $\mu_1<\mu_2<\mu_3$ against all alternatives based on 3 random variables $Y_1,Y_2$ and $Y_3$ drawn independently from $\mathcal{N}(\mu_i,1)$ for $i=1,2,3$. Note that testing $\mu_1\leq\mu_2$ is equivalent to testing $\mu_2-\mu_1\leq 0$ which we discussed in paragraph \ref{subsec:OneVarTypeB}.\\
Let $\bar{Y}$ be the constrained maximum likelihood that we calculate using the pool adjacent violators algorithm or the PAVA (\citet{BartholomewPAVA}, \citet{vanEeden}). Define $r(\bar{Y})$ as the number of equal coordinates (the number of pooled ones). Note that $r=0$ if $Y_1<Y_2<Y3$, and $r=2$ if $Y_3<Y_1<Y_2$. Define the critical random value as follows
\[ q(Y,\alpha) = q_{r(\bar{Y})}.\]
We need to show that
\[\mathbb{P}_{H_1}\left(LR>q(Y,\alpha)\right)\leq \alpha\]
According to the relative positions of the $Y$'s, the number of terms in the LR changes. Denote $S_i$ for $i=0,\cdots,7$ these cases with $S_0=\{y\in\mathbb{R}^3: y_1\leq y_2\leq y_3\}$. We have
\[\mathbb{P}_{H_1}\left(LR>q(Y,\alpha)\right) = \sum_{i=1}^{3}{\mathbb{P}_{H_1}\left(LR>q_{r(\bar{Y})}| Y\in S_i\right)}\mathbb{P}_{H_1}\left(Y\in S_i\right)\]
The objective is to show that each one of these probabilities is maximized when $Y$ has a mean equal to zero. The LR is given by
\[ LR = 
 \left\{\begin{array}{rcl}
0 & \text{if} & Y_1<Y_2<Y_3  \\
(Y_1-Y_2)^2 & \text{if} & \{Y_2< Y_1<Y_3\} \text{ or } \{ Y_2< Y_3<Y_1,\; (Y_1+Y_2)/2<Y_3\} \\
(Y_2-Y_3)^2/2 & \text{if} & \{Y_1<Y_3<Y_2\} \text{ or } \{Y_3<Y_1<Y_2,\; (Y_2+Y_3)/2>Y_1\} \\
\sum_{i=1}^3{(Y_i-\hat{\mu})^2} & \text{if} & \{Y_3<Y_2<Y_1\} \text{ or } \{Y_2< Y_3<Y_1,\; (Y_1+Y_2)/2>Y_3\} \\
 &  & \text{ or } \{Y_3<Y_1<Y_2,\; (Y_2+Y_3)/2<Y_1\}\\
\end{array}\right.
\]
Define $U_1 = (Y_1-Y_2)/\sqrt{2}$ and $U_2 = (Y_1+Y_1-2Y_3)/\sqrt{6}$. Define also $\nu_1 = (\mu_1-\mu_2)/\sqrt{2}$ and $\nu_2 = (\mu_1+\mu_2-2\mu_3)/\sqrt{6}$. Note that $U_i\sim\mathcal{N}(\nu_i,1)$ for $i=1,2$ and $U_1$ is independent from $U_2$. Note also that (\citet[Theorem 5]{Apostol})
\begin{align*}
\sum_{i=1}^3{(Y_i-\hat{\mu})^2} = & \frac{1}{2}(Y_1-Y_2)^2 + \frac{2}{3}\left(Y_3 - \frac{Y_1+Y_2}{2}\right)^2 \\
 = & U_1^2 + U_2^2
\end{align*}
The LR can be rewritten as
\[ LR = 
 \left\{\begin{array}{rcl}
0 & \text{if} & U_1<0,\;\sqrt{3}U_2-U_1<0 \\
U_1^2 & \text{if} &  U_1>0,\; U_2>0\\
\frac{1}{4}\left(\sqrt{3}U_2-U_1\right)^2 & \text{if} & \sqrt{3}U_2-U_1>0,\; \sqrt{3}U_2+3U_1<0 \\
U_1^2+U_2^2 & \text{if} & U_1>0,\; \sqrt{3}U_2+3U_1>0\\
\end{array}\right.
\]
Thus,
\begin{align*}
\mathbb{P}_{H_1}\left(LR>q(Y,\alpha)\right) = & \mathbb{P}_{H_1}\left(U_1^2>q_1\left| U_1>0, U_2>0\right.\right)\mathbb{P}_{H_1}\left(U_1>0, U_2>0\right) + \\
 & \mathbb{P}_{H_1}\left(\frac{1}{4}\left(\sqrt{3}U_2-U_1\right)^2 >q_1\left| \sqrt{3}U_2-U_1>0,\; \sqrt{3}U_2+3U_1<0\right.\right) \times \\
& \qquad \quad \mathbb{P}_{H_1}\left(U_2>0,\; \sqrt{3}U_2+3U_1<0\right) + \\
 & \mathbb{P}_{H_1}\left(U_1^2+U_2^2>q_2\left| U_1>0,\; \sqrt{3}U_2+3U_1>0\right.\right)\mathbb{P}_{H_1}\left(U_1>0,\; \sqrt{3}U_2+3U_1>0\right)
\end{align*}
Note that $\sqrt{3}U_2-U_1$ is independent from $\sqrt{3}U_2+3U_1$. Thus,
\begin{align*}
\mathbb{P}_{H_1}\left(LR>q(Y,\alpha)\right) = & \mathbb{P}_{H_1}\left(U_1^2>q_1\left| U_1>0\right.\right)\mathbb{P}_{H_1}\left(U_1>0, U_2>0\right) + \\
 & \mathbb{P}_{H_1}\left(\frac{1}{4}\left(\sqrt{3}U_2-U_1\right)^2 >q_1\left| \sqrt{3}U_2-U_1>0\right.\right)\mathbb{P}_{H_1}\left(U_2>0,\; \sqrt{3}U_2+3U_1<0\right) + \\
 & \mathbb{P}_{H_1}\left(U_1^2+U_2^2>q_2\left| U_1>0,\; \sqrt{3}U_2+3U_1>0\right.\right)\mathbb{P}_{H_1}\left(U_1>0,\; \sqrt{3}U_2+3U_1>0\right)
\end{align*}
We state that (Lemma \ref{Lemm:ConditionalIncrease}) the previous conditional probabilities are maximized when $\{\nu_1=0\}, \{\sqrt{3}\nu_2-\nu_1=0\}, \{\nu_1,\nu_2=0\}$ respectively. In other words
\[\mathbb{P}_{H_1}\left(U_1^2>q_1| U_1>0\right) \leq \mathbb{P}_{H_1}\left((U_1-\nu_1)^2>q_1\left.| U_1-\nu_1>0\right.\right) \]
\begin{multline*}
\mathbb{P}_{H_1}\left(\frac{1}{4}\left(\sqrt{3}U_2-U_1\right)^2 >q_1\left| \sqrt{3}U_2-U_1>0\right.\right) \leq \\ \mathbb{P}_{H_1}\left(\frac{1}{4}\left(\sqrt{3}(U_2-\nu_2)-(U_1-\nu_1)\right)^2 >q_1 \left| \sqrt{3}(U_2-\nu_2)-(U_1-\nu_1)>0\right.\right) 
\end{multline*}
\begin{multline*}
\mathbb{P}_{H_1}\left(U_1^2+U_2^2>q_2\left| U_1>0,\; \sqrt{3}U_2+3U_1>0\right.\right) \leq \\ \mathbb{P}_{H_1}\left((U_1-\nu_1)^2+(U_2-\nu_2)^2>q_2\left| U_1-\nu_1>0,\; \sqrt{3}(U_2-\nu_2)+3(U_1-\nu_1)>0\right.\right).
\end{multline*}
Since the norm of a standard Gaussian random variable is independent from its direction, then all the conditional probabilities in the previous display become less than $\alpha$. Thus
\begin{align*}
\mathbb{P}_{H_1}\left(LR>q(Y,\alpha)\right) \leq & \alpha \sum_{i=1}^3\mathbb{P}_{H_1}\left(Y\in S_i\right) \\
 \leq & \alpha(1-\mathbb{P}_{H_1}(Y\in S_0)) \\
 \leq & \frac{5}{6}\alpha
\end{align*}
We could use the remaining $\frac{1}{6}\alpha$ in in order to gain more power. Therefore, it suffices to define the critical value function as
\[\tilde{q}(Y,\alpha) = q\left(Y,\frac{6}{5}\alpha\right)\]
The difference in power between our new method (with $\tilde{q}(Y,\alpha)$) and the classical method which uses mixtures is illustrated in figure (\ref{fig:PowerDiffIsotonic}). Since the position of the means does not matter in this problem but rather the relative positions of the centers with respect to each other, we only vary the differences $\eta_1=\mu_1-\mu_2$ and $\eta_2=\mu_2-\mu_3$. The contours show that in the regions where the adaptive critical value is used, the gain reaches $11\%$ whereas in the regions where the mixture quantile is used, the gain reaches only $7\%$. 
\begin{figure}[ht]
\centering
\includegraphics[scale = 0.65,clip,trim=2cm 2cm 2cm 1.6cm]{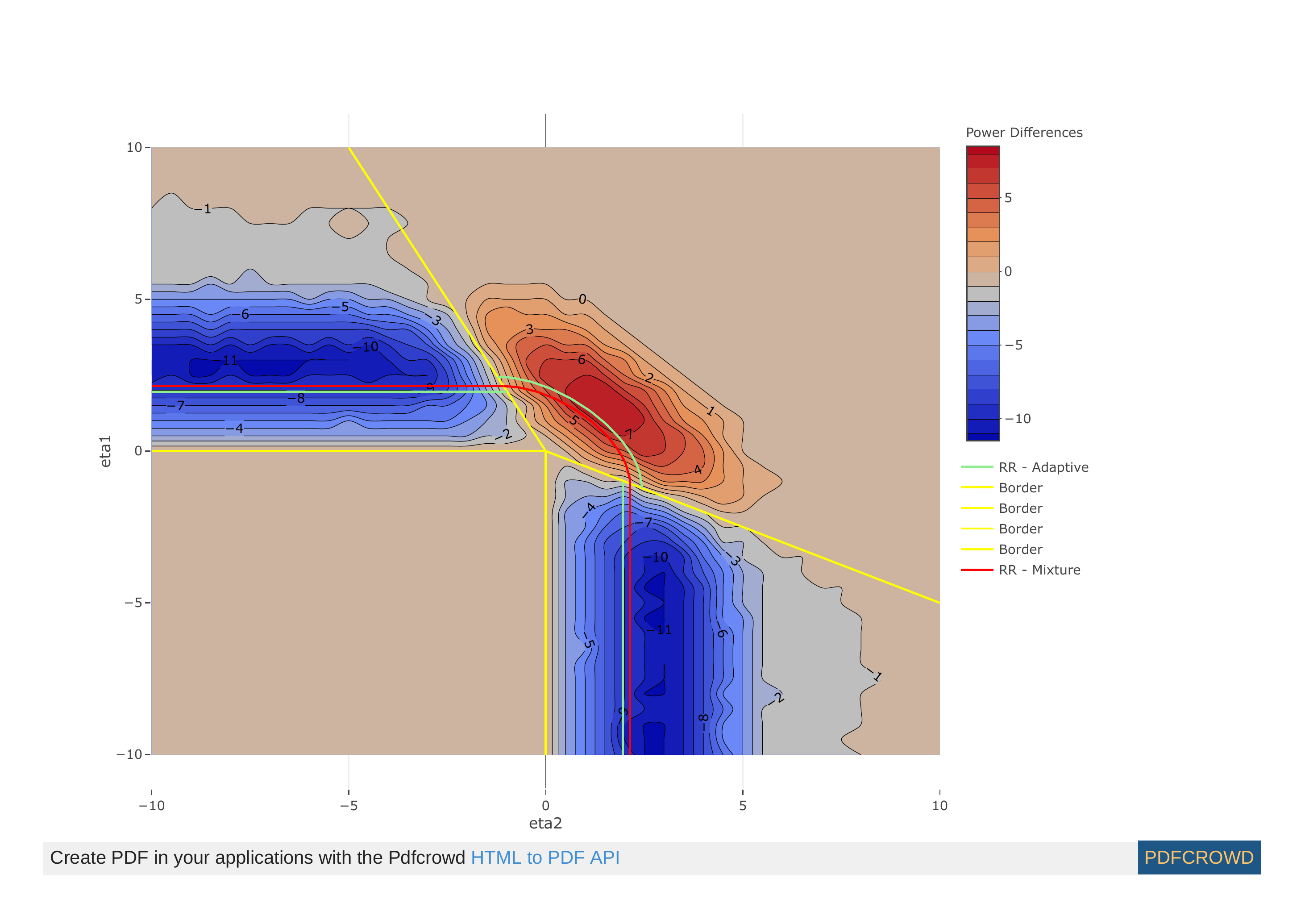}
\caption{Contours for the difference in power $100\times$(Classical$-$New). Figure produced using package \texttt{Plotly}.}
\label{fig:PowerDiffIsotonic}
\end{figure}

\subsection{The general case}\label{sec:typeBGeneral}
Let $\C\subset\mathbb{R}^p$ be some polyhedral cone. We test $H_1: \theta\in \C$ against $H_2:\theta\notin\C$. The likelihood ratio statistic is given by (\ref{eqn:LRTypeB}) which can also be expressed as (\citet[Proposition 3.4.1]{SilvapulleBook})
\[LR(\C) = \|Y - \C\|_{V} = \|\Pi_V\left(Y|\C^{o}\right)\|^2\]
where $\C^{o}$ is the polar cone corresponding to $\C$. By writing the LR as a projection over the polar cone, we could characterize this projection according to on which face it is projected. Lemma 3.13.5 from \cite{SilvapulleBook} states that when $\C^{o}$ is a polyhedral cone, there exists a collection of faces of $\C^{o}$, say $\{F_1,\cdots,F_K\}$, such that the collection of their relative interiors, $\{ri(F_1),\cdots,ri(F_K)\}$, forms a partitioning of $\C^{o}$. Furthermore, 
\begin{equation}
\|\Pi\left(Y|\C^{o}\right)\|^2 = \sum_{i=1}^K{\ind{\Pi\left(Y|\C^{o}\right)\in ri(F_i)}\|P_i Y\|^2}
\label{eqn:ProjDecompCone}
\end{equation}
where $P_i$ is the projection matrix onto the linear space spanned by $F_i$. The lemma says even more. When $\Pi\left(Y|\C^{o}\right)\in ri(F_i)$, then $\Pi\left(Y|\C^{o}\right) = P_i Y$. Define the random critical value
\begin{equation}
q(\C,Y,\alpha) = \sum_{i=1}^K{\ind{\Pi\left(Y|\C^{o}\right)\in ri(F_i)}q_{\chi^2(r_i)}(1-\alpha)}
\label{eqn:SelectiveQuantile}
\end{equation}
where $r_i = rank(P_i)$. We test $H_1:\theta\in\C$ against all alternatives using the rejection region
\begin{equation}
\left\{\|\Pi_V\left(Y|\C^{o}\right)\|^2>q(\C,Y,\alpha) \right\}
\label{eqn:RRTypeB}
\end{equation}
Our procedure consists in finding which face $F_i$ has the projection of $Y$ on $\C^{o}$, and then calculate the rank $r_i$ of the projection matrix on span$(F_i)$. The LR is then tested against the quantile of one and only one chi-squared random variable with $r_i$ degrees of freedom. The R package \texttt{coneProj} (\citet{coneProj}) calculates the projection over a polyhedral cone (and thus the LR) and provides without any extra cost the rank $r_i$ that we need for the quantile.\\
We proceed now to prove that the test defined through (\ref{eqn:RRTypeB}) has a significance level equal to $\alpha$. The following lemma proves that we have a conditional least favorable distribution according to the position of $Y$ with respect to the polar cone $\C^o$.
\begin{lemma}\label{Lemm:ConditionalIncrease}
Let $Z$ be a Gaussian random vector $\mathcal{N}(0,I_n)$. Let $q>0$, $F\subset\mathbb{R}^n$ be a cone, and let $F^o$ be its polar cone. Then,
\[\mathbb{P}\left(\|Z+\mu\|>q|Z+\mu\in F\right)\leq \mathbb{P}\left(\|Z\|>q|Z\in F\right), \forall \mu\in F^o.\]
\end{lemma}
\begin{proof}
We rewrite the conditional probability in polar coordinates. Let $S=\{\omega\in\mathbb{R}^n| \; \|\omega\|=1\}$ be the unit sphere in $\mathbb{R}^n$. Denote $\nu$ the surface measure on $S$. We have
\[\mathbb{P}\left(\|Z+\mu\|>q|Z+\mu\in F\right) = \frac{\int_{r=q}^{\infty}\int_{\omega\in S \cap F}t^{n-1}e^{-\frac{1}{2}\|tw-\mu\|^2}\nu(d\omega)dt}{\int_{r=0}^{\infty}\int_{\omega\in S \cap F}s^{n-1}e^{-\frac{1}{2}\|sw-\mu\|^2}\nu(d\omega)ds} \]
This means that the conditional density of $\|Z+\mu\|$ provided that $\{Z+\mu\in F\}$ is given by
\[g_{\mu}(t) = \frac{t^{n-1}e^{-\frac{1}{2}t^2}\int_{\omega\in S \cap F}e^{tw^T\mu}\nu(d\omega)}{\int_{r=0}^{\infty}\int_{\omega\in S \cap F}s^{n-1}e^{-\frac{1}{2}s^2+sw^T\mu}\nu(d\omega)ds}.\]
Denote $A(\mu)$ the normalization constant in the previous display, that is
\[A(\mu) = \int_{r=0}^{\infty}\int_{\omega\in S \cap F}s^{n-1}e^{-\frac{1}{2}s^2+sw^T\mu}\nu(d\omega)ds.\]
Define
\[G_{\mu}(t) = \frac{g_0(t)}{g_{\mu}(t)} = \frac{A(\mu)}{A(0)}\frac{\int_{\omega\in S \cap F}\nu(d\omega)}{\int_{\omega\in S \cap F}e^{t\omega^T\mu}\nu(d\omega)}.\]
Since $\mu\in F^o$, then 
\[\forall\omega\in S \cap F, \quad \omega^T\mu \leq 0.\]
Thus, function $t\mapsto G_{\mu}(t)$ is nondecreasing over $(0,\infty)$ whatever the value of $\mu\in F^o$.\\
Let $h$ be some measurable nondecreasing function defined on $\mathbb{R}_+$. Note that for any couple of nonnegative real numbers $(r_1,r_2)$, since both $G_{\mu}(t)$ and $h(t)$ are nondecreasing functions in $t$ over $(0,\infty)$, we have
\begin{equation}
(h(r_1)-h(r_2))(G_{\mu}(r_1)-G_{\mu}(r_2))\geq 0
\label{eqn:IncreaseDiff}
\end{equation}
Let $R_1$ and $R_2$ be two i.i.d. random variables with a common density defined on $(0,\infty)$. We have, due to (\ref{eqn:IncreaseDiff}),
\[\mathbb{E}\left[\left(h(R_1)-h(R_2)\right)\left(G(R_1)-G(R_2)\right)\right]\geq 0.\]
We deduce that
\begin{equation}
\mathbb{E}\left[h(R_1)G_{\mu}(R_1)\right]\geq \mathbb{E}\left[h(R_1)\right]\mathbb{E}\left[G_{\mu}(R_1)\right].
\label{eqn:PositiveCov}
\end{equation}
Assume now that $R_1$ has the density $g_{\mu}$, and denote $\mathbb{E}_{\mu}$ (resp. $\mathbb{E}_0$) the expectation under $g_{\mu}$ (resp. $g_0$). We have now
\begin{align*}
\mathbb{E}_{\mu}\left[h(R_1)G_{\mu}(R_1)\right] & =  \mathbb{E}_{0}\left[h(R_1)\right] \\
\mathbb{E}_{\mu}\left[G_{\mu}(R)\right] & =  1.
\end{align*}
The second line in the previous display comes from the fact that $g_0$ is a density over $(0,\infty)$. Hence, due to (\ref{eqn:PositiveCov}), we may write
\begin{equation}
\mathbb{E}_{\mu}\left[h(R)\right] \leq  \mathbb{E}_{0}\left[h(R)\right].
\label{eqn:IncreaseExpect}
\end{equation}
Set $h(t) = \ind{t>q}$ (the indicator function of the set $(q,\infty)$). Function $h$ is nondecreasing over $(0,\infty)$, therefore we could apply (\ref{eqn:IncreaseExpect}) on it and get 
\[\frac{\int_{r=q}^{\infty}\int_{\omega\in S \cap F}t^{n-1}e^{-\frac{1}{2}\|tw^T\mu\|^2}\nu(d\omega)dt}{\int_{r=0}^{\infty}\int_{\omega\in S \cap F}s^{n-1}e^{-\frac{1}{2}\|sw^T\mu\|^2}\nu(d\omega)ds} \leq \frac{\int_{r=q}^{\infty}\int_{\omega\in S \cap F}t^{n-1}\nu(d\omega)dt}{\int_{r=0}^{\infty}\int_{\omega\in S \cap F}s^{n-1}\nu(d\omega)ds}. \]
This is exactly what the lemma claims.
\end{proof}
The previous lemma is the corner stone to proving that our new test has a correct $\alpha$ level. We present now our main result.
\begin{proposition}[type B]\label{prop:typeBGeneral}
Let $Y\sim \mathcal{N}_p(\mu,V)$ and $\C$ be a polyhedral cone. A valid test at level $\alpha$ for $H_1:\theta\in\C$ against $\theta\notin\C$ is given by the rejection region
\begin{equation}
\left\{LR(\C,V)>q(\C,Y,\alpha)\right\}
\label{eqn:SelectiveTest}
\end{equation}
where $q(\C,Y,\alpha)$ is given by (\ref{eqn:SelectiveQuantile}). Moreover,
\[\mathbb{P}_{H_1}\left(LR(\C)>q(\C,Y,\alpha)\right)\leq \left(1-\mathbb{P}_{H_1}\left(Y\in\C\right)\right)\alpha\]
\end{proposition}
\begin{proof}
We start with the case when $V=I_p$. Due to (\ref{eqn:ProjDecompCone}), we may write
\begin{align*}
\mathbb{P}_{H}\left(LR(\C)>q(\C,Y,\alpha)\right) =& \sum_{j=1}^K\mathbb{P}_{H_1}\left(\left\{\|Y-\C\|^2>q_{r_j}\right\}\cap \left\{\Pi\left(Y|\C^{o}\right)\in ri(F_j)\right\}\right) \\
= & \sum_{j=1}^K\mathbb{P}_{H_1}\left(\|Y-\C\|^2>q_{r_j} \left| \Pi\left(Y|\C^{o}\right)\in ri(F_j)\right.\right)\mathbb{P}_{H_1}\left(\Pi\left(Y|\C^{o}\right)\in ri(F_j)\right)
\end{align*}
We treat each of the conditional probabilities separately. Lemma 3.13.2 from \cite{SilvapulleBook} states that the event
\[\left\{\Pi\left(Y|\C^{o}\right)\in ri(F_j)\right\}\]
is equivalent to the event
\[\{P_jY\in ri(F_j),\; (I-P_j)Y\in F_j^{\bot}\cap \C^o\}\]
Recall that $\|Y-\C\|^2 = \Pi\left(Y|\C^{o}\right)$. Moreover, when $\Pi\left(Y|\C^{o}\right)\in ri(F_j)$, then $\Pi\left(Y|\C^{o}\right) = P_j Y$, thus
\begin{align*}
\mathbb{P}_H\left(\|Y-\C\|^2>q_{r_j}\left| \Pi\left(Y|\C^{o}\right)\in ri(F_j)\right.\right) = &
\mathbb{P}_H\left(\|P_jY\|^2>q_{r_j}\left| \Pi\left(Y|\C^{o}\right)\in ri(F_j)\right.\right) \\
= & \mathbb{P}_H\left(\|P_jY\|^2>q_{r_j}\left| P_jY\in ri(F_j), (I-P_j)Y\in F_j^{\bot}\cap \C^o\right.\right)
\end{align*}
Because $P_j$ is a projection matrix, then the random variables $P_jY$ and $(I-P_j)Y$ are independent. Thus,
\[
\mathbb{P}_{H_1}\left(\|Y-\C\|^2>q_{r_j}\left| \Pi\left(Y|\C^{o}\right)\in ri(F_j)\right.\right) = \mathbb{P}_{H_1}\left(\|P_jY\|^2>q_{r_j}\left| P_jY\in ri(F_j)\right.\right) \]
We apply Lemma \ref{Lemm:ConditionalIncrease} on $Z=P_jY-P_j\mu$, $\C=P_j\C$. Note that for all $\mu\in\C$ and $y\in ri(F_j)$, we have
\begin{align*}
y^TP_j\mu = & (P_j^Ty)^T\mu\\
 = & (P_jy)^T\mu \\
= & y^T\mu \\
 \leq & 0
\end{align*}
\begin{figure}[ht]
\centering
\includegraphics[scale=0.5]{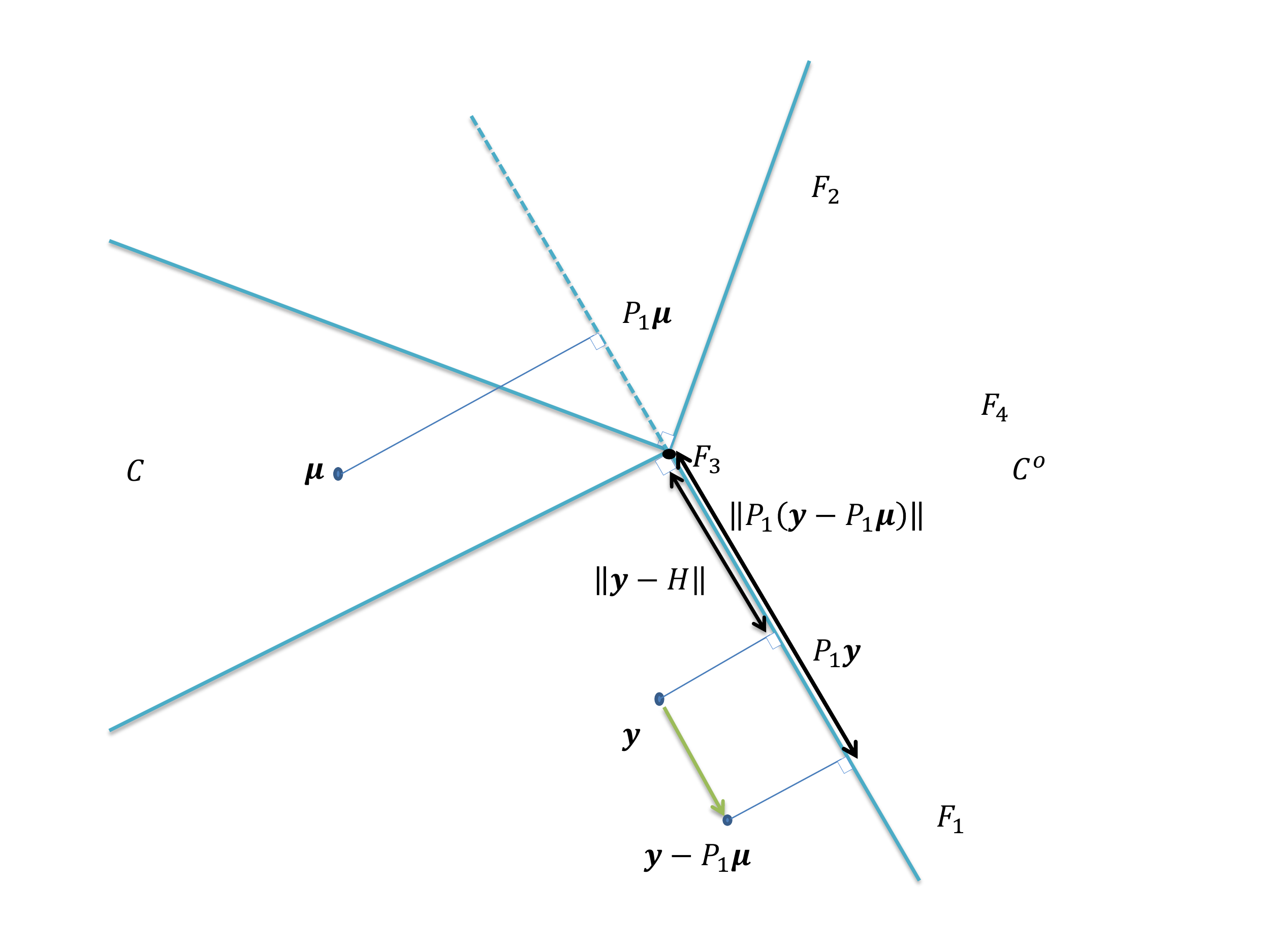}
\caption{The conditional least favorable distribution. The polar cone has 4 faces, F3 is the intersection point of $F_1$ and $F_2$, whereas $F4=\C^o$.}
\label{fig:LeastFav}
\end{figure}
The second line comes from the fact that $P_j$ is a projection matrix so that it is symmetric, and the third lines is because $P_j$ is the projection matrix onto $span(F_j)$ and $y\in ri(F_j)$. Thus, for any $\mu$ in $\C$, $P_j\mu$ is in the polar cone of $F_j$. See figure (\ref{fig:LeastFav}) for an illustration. Thus,
\[\mathbb{P}_{H_1}\left(\|P_jY\|^2>q_{r_j}\left| P_jY\in ri(F_j)\right.\right)\leq \mathbb{P}_{H_1}\left(\|P_jY-P_j\mu\|^2>q_{r_j}\left| P_jY-P_j\mu\in ri(F_j)\right.\right)\]
Now since $Y-\mu\sim\mathcal{N}(0,I_p)$, then according to Lemma 3.13.3 from \cite{SilvapulleBook}, the distribution of $\|P_j(Y-\mu)\|^2=\|P_jY-P_j\mu\|^2$ conditionally on $P_jY-P_j\mu\in\C$ is $\chi^2(r_j)$. Thus,
\begin{align*}
\mathbb{P}_{H_1}\left(\|P_jY\|^2>q_{r_j}\left| P_jY\in ri(F_j)\right.\right)\leq & \mathbb{P}_{H_1}\left(\|P_jY-P_j\mu\|^2>q_{r_j}\left| P_jY-P_j\mu\in ri(F_j)\right.\right) \\
\leq & \alpha.
\end{align*}
Finally,
\begin{align*}
\mathbb{P}_{H}\left(LR(\C)>q(\C,Y,\alpha)\right) =& \sum_{j=1}^K\mathbb{P}_{H_1}\left( \|\Pi_{\C^o}Y\|^2>q_{r_j}\left| \Pi\left(Y|\C^{o}\right)\in ri(F_j)\right.\right)\mathbb{P}_{H_1}\left(\Pi\left(Y|\C^{o}\right)\in ri(F_j)\right) \\
\leq & \alpha\sum_{j=1}^K\mathbb{P}_{H_1}\left(\Pi\left(Y|\C^{o}\right)\in ri(F_j)\right)\\
\leq & \left(1-\mathbb{P}_{H_1}\left(\Pi\left(Y|\C^{o}\right)\in ri(F_0)\right)\right)\alpha\\
\leq & \left(1-\mathbb{P}_{H_1}\left(Y\in\C\right)\right)\alpha \\
< & \alpha
\end{align*}
We have proved that (\ref{eqn:SelectiveTest}) is a valid test at level $\alpha$ when $V=I_p$. We move now to the general case for any symmetric positive definite matrix $V$. Let $V^{-1}=A^TA$, and $Z=AY$, then $Z$ is distributed as $\mathcal{N}_p(A\mu,I)$ and
\[LR = \|Z-A\C\|^2.\]
We use the same arguments above on $Z$ and the polyhedral cone $A\C$ instead of $Y$ and $\C$ respectively. Note that the face of $A\C$ takes the form of $AF$ for some face of $\C$ and that dim$(span(F_i))=$ dim$(span(AF_i))$ (\citet[p. 129]{SilvapulleBook}). Therefore, since the degrees of freedom in $q(A\C,Y,\alpha)$ depend only on the rank of the projection matrix onto the linear space spanned by $AF$, then they stay the same whether they are calculated for $Y$ or for $Z$ so that $q(\C,Y,\alpha)=q(A\C,Z,\alpha)$. Thus,
\begin{align*}
\mathbb{P}_{H}\left(LR(V,\C)>q(\C,Y,\alpha)\right) = & \mathbb{P}_{H}\left(LR(A\C)>q(A\C,X,\alpha)\right) \\
 \leq & \left(1-\mathbb{P}_{H_1}\left(Z\in A\C\right)\right)\alpha\\
< & \alpha
\end{align*}
\end{proof}
Whenever it is possible to calculate the term $\mathbb{P}_{H_1}\left(Y\in \C\right)$ or at least provide an upper bound, then one could use it to adjust the critical level for the chi-squares quantiles and gain more power. Two examples are discussed hereafter. 
\begin{remark}
In the literature, the test (\ref{eqn:RRTypeB}) in the case of an isotonic hypothesis was named as a conditional likelihood ratio test (\citet{WollanDykstra,IversonHarp}). We do not really agree on that name. The fact that the critical value is random is not classical, however we could look at the test differently. The classical way of constrained LRT tests the LR against a quantile of a mixture, say $c(\alpha)$. The rejection region is
\[\{\left\{\|\Pi_V\left(Y|\C^{o}\right)\|^2 - c(\alpha)>0\right\}.\]
The rejection region (\ref{eqn:RRTypeB}) could be rewritten as
\[\left\{\|\Pi_V\left(Y|\C^{o}\right)\|^2-q(\C,Y,\alpha)>0 \right\}\]
In other words, our new approach appears as if the test statistic was changed from $LR(\C)-c(\alpha)$ into $LR(\C)-q(\C,Y,\alpha)$. The "critical value" for both tests is the same and is 0.
\end{remark}
\section{Examples}\label{sec:Examples}
\subsection{The special case of orthant hypotheses revisited}
Let $y_1,\cdots,y_p$ be $n$ i.i.d. realizations from the Gaussian distributions $\mathcal{N}(\mu_i,1)$ for $i=1,\cdots,p$. When the null hypothesis is $H_1:\mu_1\leq 0, \cdots,\mu_p\leq 0$, then the polyhedral $\C$ from Proposition \ref{prop:typeBGeneral} is $\{y\in\mathbb{R}^p: -I_py\geq 0\}$. The LR is given by
\[LR = \sum_{i=1}^n{y_i^2\ind{y_i\geq 0}}.\]
The quantile function $q(Y,\alpha)$ is given by
\[q(Y,\alpha) = q_{\chi^2(N)}, \qquad N = \sum_{i=1}^n{\ind{Y_i\geq 0}}.\]
We could still gain a little more power by considering the immediate lower bound on the adjustment of the critical level. 
\begin{align*}
\mathbb{P}_{H_1}\left(Y\in \C\right) &\geq \mathbb{P}_{0}\left(Y\in \C\right) \\
 & \geq 2^{-p}
\end{align*}
Thus, in order to test the orthant $H_1$ at level $\alpha$, we use the quantile $q(Y,\bar{\alpha})$ defined by
\[q(Y,\bar{\alpha}) = q_{\chi^2(N)}\left(1-\frac{2^p}{2^p-1}\alpha\right)\]
For small values of $p$, the gain from this adjustment becomes important where it almost vanishes as the dimension $p$ grows. 
\subsection{The special case of isotonic hypotheses revisited: Maybe study a hypothesis with equalities!}
When the null hypothesis is $H_1:\mu_1\leq\cdots\leq\mu_p$, then the polyhedral $\C$ from Proposition \ref{prop:typeBGeneral} is $\{y\in\mathbb{R}^p: Ay\geq 0\}$ where the matrix $A$ has all its elements zero except for $A_{i,i} = -1$ and $A_{i,i+1}=1$ for all $i=1,\cdots,p-1$. The likelihood ratio can be calculated easily using the PAVA (\citet{BartholomewPAVA, vanEeden}). Function \texttt{isoreg} in the statistical program R does the job. If the number of levels (distinct values) in the result of the PAVA is $l$, then the LR is tested against the quantile of a $\chi^2(n-l)$ at order $1-\alpha$. We could adjust the order of the quantile in order to gain more power. We have the following immediate lower bound on the adjustment term (\citet[Corollary 2.6]{Robertson78})
\begin{align*}
\mathbb{P}_{H_1}\left(Y\in \C\right) &\geq  \mathbb{P}_{0}\left(Y\in \C\right) \\
 & \geq \frac{1}{p!}
\end{align*}
Thus, we can test the LR against the quantile of $\chi^2(n-l)$ at order $1-\frac{p!}{p!-1}\alpha$. \cite{WollanDykstra} proposed in the case of the isotonic hypothesis to give an estimate instead of an lower bound for the adjustment term $\mathbb{P}_{H_1}\left(Y\in \C\right)$ so that we gain even more power. They state however, that it could lead sometimes to increase the significance level of the test more than $\alpha$.\\


\section{The case of unknown variances}\label{sec:UnKnownVar}
Let $Y\sim\mathcal{N}(\mu,V)$. Assume that $V=\sigma^2\Sigma$ for some known matrix $\Sigma$ and unknown $\sigma^2$. Following \cite{Kudo}, assume that we dispose of an estimator for $\sigma^2$, say $\hat{\sigma}^2$, such that $\hat{\sigma}^2/\sigma^2$ is distributed independently as $\chi^2(m)$. In this section, $c_{j,m}$ is the quantile of an F-distributed random variable with degrees of freedom $j$ and $m$. We test $H_1:\mu\in\C$ against all alternatives. Let $X\sim\mathcal{N}(\mu/\sigma^2,\Sigma)$. Note that
\begin{align*}
LR(Y) = & \min_{\theta\in\C} (Y-\theta)^T\hat{V}^{-1}(Y-\theta)\\
 = & \min_{\theta\in\C} \frac{1}{\hat{\sigma}^2}(\sigma X-\theta)^T\Sigma^{-1}(\sigma X-\theta)\\
 = & \frac{\sigma^2}{\hat{\sigma}^2}\min_{\theta\in\C} (X-\theta)^T\Sigma^{-1}(X-\theta)\\
= & \frac{\sigma^2}{\hat{\sigma}^2} LR(X)
\end{align*}
The third line comes from the fact that $\C$ is a cone and that $\sigma^2>0$ so that if $\mu\in\C$ then $\sigma^2\mu\in\C$. Define the random critical function as follows
\[q(\C,Y,\alpha) = \sum_{i=1}^K{\ind{\Pi_{\hat{V}}\left(Y|\C^{o}\right)\in ri(F_i)}c_{r_i,m}}(1-\alpha)\]
where the $F_j$'s are, as in Section \ref{sec:typeBGeneral}, the faces of the polar of the polyhedral cone $\C^o$ and they form a partitioning of it.\\
We want to show that
\[\mathbb{P}_{H_1}(LR(Y))>q(\C,Y,\alpha) \leq \alpha.\]
\begin{proof}
We start with the case when $\Sigma=I_p$. Due to (\ref{eqn:ProjDecompCone}), we may write
\begin{align*}
\mathbb{P}_{H}\left(LR(\C)>q(\C,Y,\alpha)\right) =& \sum_{j=1}^K\mathbb{P}_{H_1}\left(\frac{\sigma^2}{\hat{\sigma}^2}\|X-\C\|^2>c_{r_j,m} \left| \Pi\left(X|\C^{o}\right)\in ri(F_j)\right.\right)\mathbb{P}_{H_1}\left(\Pi\left(X|\C^{o}\right)\in ri(F_j)\right)
\end{align*}
Using Lemma 3.13.2 from \cite{SilvapulleBook} we may write
\begin{align*}
\mathbb{P}_H\left(\frac{\sigma^2}{\hat{\sigma}^2}\|X-\C\|^2>c_{r_j,m}\left| \Pi\left(X|\C^{o}\right)\in ri(F_j)\right.\right) = &
\mathbb{P}_H\left(\frac{\sigma^2}{\hat{\sigma}^2}\|P_jX\|^2>c_{r_j,m}\left| \Pi\left(X|\C^{o}\right)\in ri(F_j)\right.\right) \\
= & \mathbb{P}_H\left(\frac{\sigma^2}{\hat{\sigma}^2}\|P_jX\|^2>c_{r_j,m}\left| P_jX\in ri(F_j), (I-P_j)X\in F_j^{\bot}\cap \C^o\right.\right)
\end{align*}
Because $P_j$ is a projection matrix, then the random variables $P_jX$ and $(I-P_j)X$ are independent. Besides, $(I-P_j)X$ is independent from $\hat{\sigma}^2$. Thus,
\begin{align*}
\mathbb{P}_{H_1}\left(\frac{\sigma^2}{\hat{\sigma}^2}\|X-\C\|^2>c_{r_j,m}\left| \Pi\left(X|\C^{o}\right)\in ri(F_j)\right.\right) =& \mathbb{P}_{H_1}\left(\frac{\sigma^2}{\hat{\sigma}^2}\|P_jX\|^2>c_{r_j,m}\left| P_jX\in ri(F_j)\right.\right) \\
 = & \int_{s>0}\mathbb{P}_{H_1}\left(s\|P_jX\|^2>c_{r_j,m}\left| P_jX\in ri(F_j)\right.\right)f_{\chi^2(m)}(s)ds
\end{align*}
We apply Lemma \ref{Lemm:ConditionalIncrease} on $Z=P_jX-\frac{1}{\sigma^2}P_j\mu$, $\C=P_j\C$ and $q=c_{r_j,m}/s$. Note that for all $\mu\in\C$ and $y\in ri(F_j)$, we have $y^TP_j\mu\leq 0$. Thus, for any $\mu$ in $\C$, $P_j\mu$ is in the polar cone of $F_j$ and so does $\frac{1}{\sigma^2}P_j\mu$ because $F_j$ is a polyhedral cone. Lemma \ref{Lemm:ConditionalIncrease} states then 
\[\mathbb{P}_{H_1}\left(s\|P_jX\|^2>c_{r_j,m}\left| P_jX\in ri(F_j)\right.\right)\leq \mathbb{P}_{H_1}\left(s\left\|P_jX-\frac{1}{\sigma^2}P_j\mu\right\|^2>c_{r_j,m}\left| P_jX-\frac{1}{\sigma^2}P_j\mu\in ri(F_j)\right.\right)\]
Integrating both sides with respect to $f_{\chi^2(m)}(s)$, we get
\[\mathbb{P}_{H_1}\left(\frac{\sigma^2}{\hat{\sigma}^2}\|P_jX\|^2>c_{r_j,m}\left| P_jX\in ri(F_j)\right.\right)\leq \mathbb{P}_{H_1}\left(\frac{\sigma^2}{\hat{\sigma}^2}\left\|P_jX-\frac{1}{\sigma^2}P_j\mu\right\|^2>c_{r_j,m}\left| P_jX-\frac{1}{\sigma^2}P_j\mu\in ri(F_j)\right.\right)\]
Now since $X-\frac{1}{\sigma^2}\mu\sim\mathcal{N}(0,I_p)$, then according to Lemma 3.13.3 from \cite{SilvapulleBook}, the distribution of $\|P_j(X-\frac{1}{\sigma^2}\mu)\|^2=\|P_jX-\frac{1}{\sigma^2}P_j\mu\|^2$ conditionally on $P_jX-\frac{1}{\sigma^2}P_j\mu\in\C$ is $\chi^2(r_j)$. Thus, $\frac{\sigma^2}{\hat{\sigma}^2}\|P_jX-\frac{1}{\sigma^2}P_j\mu\|^2$ conditionally on $P_jX-\frac{1}{\sigma^2}P_j\mu\in\C$ has the F-distribution with degrees of freedom $m$ and $r_j$. Hence,
\begin{align*}
\mathbb{P}_{H_1}\left(\frac{\sigma^2}{\hat{\sigma}^2}\|P_jX\|^2>c_{r_j,m}\left| P_jX\in ri(F_j)\right.\right)\leq & \mathbb{P}_{H_1}\left(\frac{\sigma^2}{\hat{\sigma}^2}\|P_jX-\frac{1}{\sigma^2}P_j\mu\|^2>c_{r_j,m}\left| P_jX-\frac{1}{\sigma^2}P_j\mu\in ri(F_j)\right.\right) \\
\leq & \alpha.
\end{align*}
Finally,
\begin{align*}
\mathbb{P}_{H}\left(LR(Y)>q(\C,Y,\alpha)\right) =& \sum_{j=1}^K\mathbb{P}_{H_1}\left(\frac{\sigma^2}{\hat{\sigma}^2}\|\Pi(X|\C^o)\|^2>c_{r_j,m}\left| \Pi\left(X|\C^{o}\right)\in ri(F_j)\right.\right)\mathbb{P}_{H_1}\left(\Pi\left(X|\C^{o}\right)\in ri(F_j)\right) \\
\leq & \alpha\sum_{j=1}^K\mathbb{P}_{H_1}\left(\Pi\left(X|\C^{o}\right)\in ri(F_j)\right)\\
\leq & \left(1-\mathbb{P}_{H_1}\left(\Pi\left(X|\C^{o}\right)\in ri(F_0)\right)\right)\alpha\\
\leq & \left(1-\mathbb{P}_{H_1}\left(Y\in\C\right)\right)\alpha \\
< & \alpha
\end{align*}
We have proved that (\ref{eqn:SelectiveTest}) is a valid test at level $\alpha$ when $\Sigma=I_p$. The case when $\Sigma$ is not the identity matrix is treated similarly to the proof of Proposition \ref{prop:typeBGeneral} by considering the decomposition $\Sigma^{-1}=A^TA$.
\end{proof}

\section{Power comparison}\label{sec:PowerCompare}
We illustrate in two examples when our approach overpowers the classical one by linking it to the number of violated restrictions. By definition, the power of the statistical test of type B for the classical approach is given by
\[\mathbb{P}_{H_2}(LR(Y)>q(\alpha)),\]
where $q$ is the quantile of order $1-\alpha$ of the mixture of chi-squares (\ref{eqn:ChiBarTypeB}). For our approach, the power for type B is given by
\[\mathbb{P}_{H_2}(LR(Y)>q(Y,\alpha)),\]
where $q(Y,\alpha)$ is given by (\ref{eqn:SelectiveQuantile}). In order to see when our method is more powerful than the classical one, it suffices to understand when the event $\{q(Y)>q\}$ happens and of course at what frequency does it occur. For the classical approach, the quantile $q(\alpha)$ is determined by how each of the $\chi^2(i)$ is weighted for $i=1,\cdots,n-1$. If the vector of weights gives more weights to the larger degrees of freedom (the case of isotonic hypothesis), then the quantile $q(\alpha)$ is closer to the quantile of a chi-square with a high degree of freedom and vice-versa. We believe that if the covariance matrix is $I_n$, then if the null hypothesis is strictly smaller than a quadrant, then the weights tend to give more credit for the large degrees of freedom (because the polar cone of the null becomes larger than a quadrant). This is for example the case of an isotonic hypothesis. \\

We look at two situations where the vector of weights in (\ref{eqn:ChiBarTypeB}) are calculated explicitly; namely the case of $H_1$ is a quadrant, and the case when $H_1$ is the isotonic hypothesis $\mu_1\leq\cdots\leq\mu_n$.\\
In order to find out when the event $\{q(Y,\alpha)>q(\alpha)\}$ happens, we fix $n$ and $\alpha$ and then search for the largest chi-square quantile of order $1-\alpha$ smaller than $q(\alpha)$. For example, in the isotonic situation, when $n=3$ we have $q(0.05)\approx 4.6$ and the largest chi-square quantile of order $0.95$ is the quantile of $\chi^2(1)\approx 3.8$. Thus, if in our approach, the vector $Y$ violates one restriction imposed by the null, then $q(Y,\alpha)$ is the quantile of $\chi^2(1)$. If the vector $Y$ violates the two restrictions imposed by the null, then $q(Y,\alpha)$ is the quantile of $\chi^2(2)$. In the first case, our approach is more powerful than the classical one, whereas the classical one is the more powerful approach in the second case.\\

When the null $H_1$ is a quadrant, figure (\ref{fig:MaxNbMissUpsTypeB}) shows for each $p$ from $2$ to $100$, the maximum number of constraints violations below which our new approach is more powerful than the classical one. In this case, the weights are largest for quadrants with $E(p/2)+1$ where $E$ denotes the integer part. Besides, the weights are repartitioned equally around this maximum. Thus, the quantile $q(\alpha)$ is most likely to be closest to the quantile of $\chi^2(E(p/2)+1)$. If the variance of the data is not very large, then we expect that the vector of observations will most likely respect most of the signs of the vector of means $\mu$. Therefore, as long as the vector of means $\mu$ is in one of the quadrants with a number of changes of signs less than $E(p/2)+1$, then our new approach is more likely to be more powerful than the classical approach because $q(Y,\alpha)$ will most of the time be equal to the quantile of a chi-square with a degree of freedom smaller than $E(p/2)+1$. 
\begin{figure}[ht]
\centering
\includegraphics[scale = 0.5]{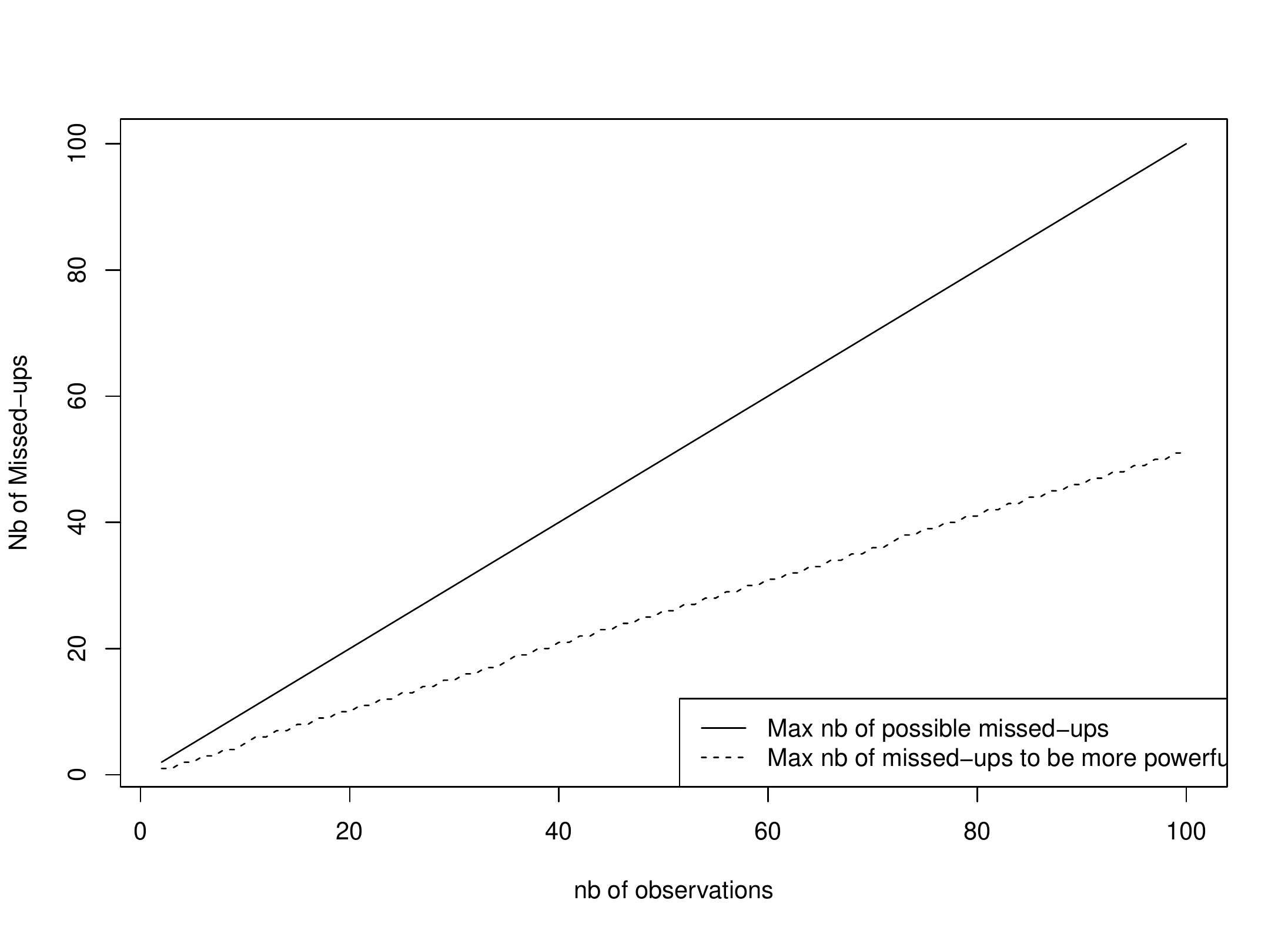}
\caption{}
\label{fig:MaxNbMissUpsTypeB}
\end{figure}
We illustrate another setup with a different polyhedral cone than an orthant. When the null $H_1$ is the isotonic hypothesis $\mu_1\leq\cdots\leq\mu_n$, figure (\ref{fig:MaxNbMissUps}) shows for each $n$ from $2$ to $100$, the maximum number of constraints violations below which our new approach is more powerful than the classical one. The figure suggests that our approach is more powerful than the classical most of the time. The cases when the classical approach overpowers our approach are the extreme cases when most of the observations violate the restrictions of the null. 
\begin{figure}[ht]
\centering
\includegraphics[scale = 0.5]{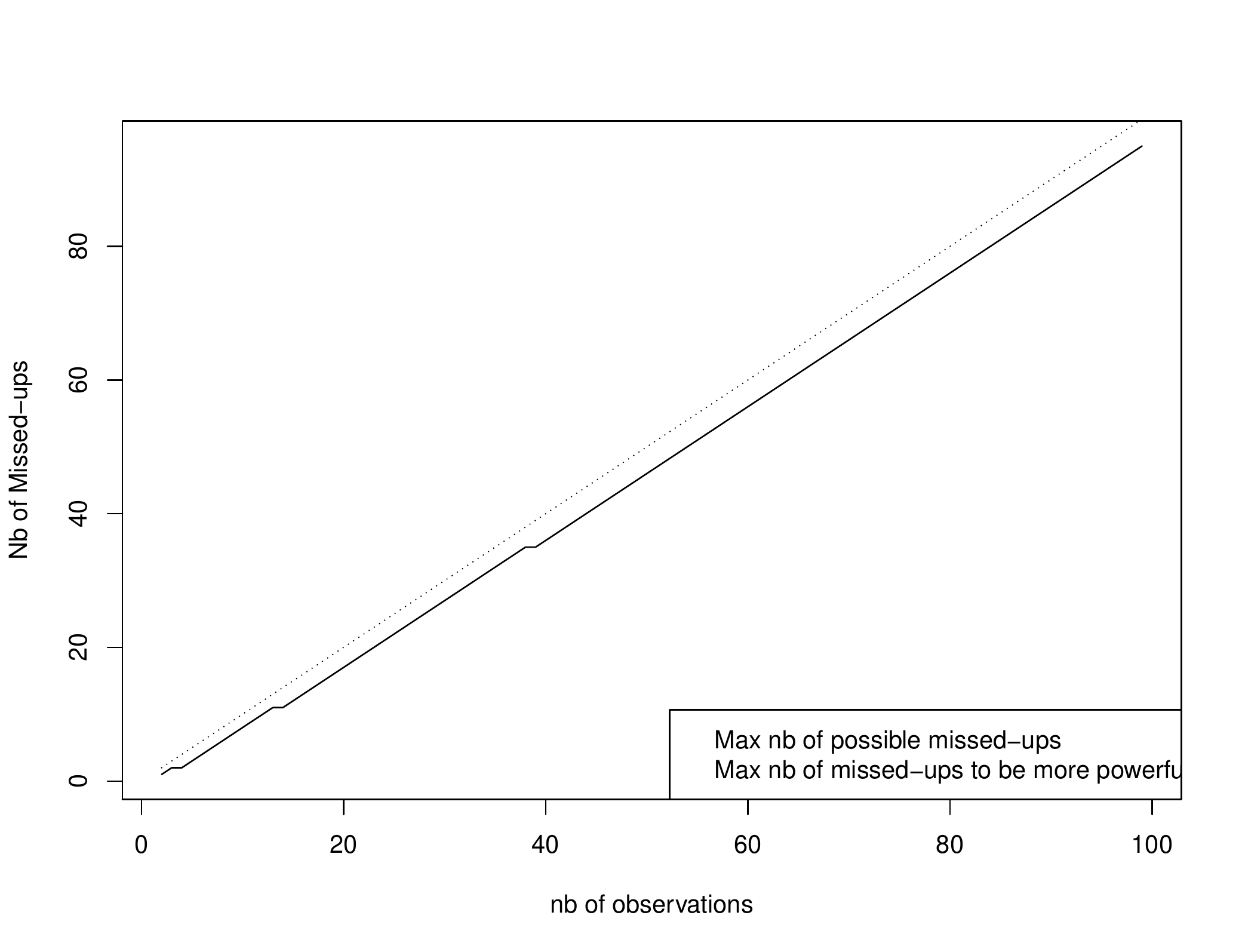}
\caption{}
\label{fig:MaxNbMissUps}
\end{figure}


\section{Discussion}
We presented in this paper a novel approach to testing ordered hypotheses which changed the whole idea of how one could look at this problem. Our novel approach was shown to be interesting in terms of both power and simplicity in comparison to the classical approach. Indeed, we are replacing the quantile of a mixture of chi-squares with the quantile of only one of them according to the data and without any extra cost. This avoided the complication of calculating the weights for the mixture and provided more power in a large part of the alternative especially those which are not very far from the null hypothesis (where power is usually small).\\
The idea of using a random critical value function whose value is selected after seeing the data is rather general and could be employed in other contexts in order to obtain simple and powerful tests. In this paper, we only considered type A and type B problems, but there are more than this when working on ordered hypotheses that could be treated anew using our approach (\citet[Chap. 4, 6, 7]{SilvapulleBook}). Future research could shed light on some of them.

\section{Appendix: An alternative proof in the the two dimensional case}\label{append1}
Let $Y\sim(\mu,I_2)$ be a bivariate Gaussian random variable. Consider the testing problem
\[H_1: \mu_1\leq 0, \mu_2\leq 0, \text{ against } H_2: \mu\in \mathbb{R}^p.\]
We are testing if the vector $\mu$ is in the negative quadrant of the plane. The likelihood ratio for this test is given by
\[LR = \left\{\begin{array}{rcl}
0 & \text{if} & Y_1\leq 0, Y_2\leq 0 \\
Y_1^2 & \text{if} & Y_1\geq 0, Y_2\leq 0 \\
Y_2^2 & \text{if} & Y_1\leq 0, Y_2\geq 0 \\
Y_1^2+Y_2^2 & \text{if} & Y_1\geq 0, Y_2\geq 0.
\end{array}\right.\]
We define the random quantile
\[ q(Y,\alpha) = 
 \left\{\begin{array}{rcl}
0 & \text{if} & Y_1\leq 0, Y_2\leq 0 \\
q_1 & \text{if} & Y_1\geq 0, Y_2\leq 0 \text{ or } Y_1\leq 0, Y_2\geq 0\\
q_2 & \text{if} & Y_1\geq 0, Y_2\geq 0.
\end{array}\right.
\]
We test $H_1$ against $H_2$ using the rejection region 
\[\left\{LR>q(Y,\alpha)\right\}.\]
Our claim is that 
\[\mathbb{P}_{H_1}\left(LR>q(Y,\alpha)\right)\leq \alpha.\]
\begin{lemma}\label{lem:ConditionalIncrease}
Let $Y\sim\mathcal{N}(\mu,1)$ with $\mu\in\mathbb{R}$, then function
\[\mu\mapsto \mathbb{P}_{\mu}\left(Y^2>q|Y\geq 0\right)\]
is increasing, and thus if $q = q_1$, then
\[\mathbb{P}_{\mu}\left(Y^2>q|Y\geq 0\right)\leq \mathbb{P}_{\mu}\left((Y-\mu)^2>q|Y-\mu\geq 0\right) = \alpha, \quad \forall \mu<0.\]
\end{lemma}
\begin{proof}
We rewrite the conditional probability
\[
\mathbb{P}_{H_1}\left(Y_1^2>q_1| Y_1\geq 0\right) = \frac{\mathbb{P}_{H_1}\left(Y_1^2>q_1\cap Y_1\geq 0\right)}{\mathbb{P}_{H_1}\left(Y_1\geq 0\right)}
\]
Set $Z = Y_1-\mu_1$, and define $\Phi(x)$ and $\varphi(x)$ as the cdf and the pdf of a standard normal distribution. We have
\begin{align*}
\mathbb{P}_{H_1}\left(Y_1^2>q_1| Y\in Q_4\right)& = \frac{\mathbb{P}_{H_1}\left(Z>-\mu_1+\sqrt{q_1}\right)}{\mathbb{P}_{H_1}\left(Z\geq -\mu_1\right)} \\
& = \frac{1-\Phi(-\mu_1+\sqrt{q_1})}{1-\Phi(-\mu_1)}
\end{align*}
We prove that the function 
\[f(x) = \frac{1-\Phi(-x+\sqrt{q_1})}{1-\Phi(-x)}\]
is increasing. Indeed, since $\log(f(x)) = \log(1-\Phi(-x+\sqrt{q_1})) - \log(1-\Phi(-x))$, its derivative with respect to $x$ is given by
\[\log(f(x))' = \frac{\varphi(-x+\sqrt{q_1})}{1-\Phi(-x+\sqrt{q_1})} - \frac{\varphi(-x)}{1-\Phi(-x)}.\]
The inverse of Mill's ratio $\frac{\varphi(a)}{1-\Phi(a)}$ is increasing (\cite{Sampford}), so that
\[\frac{\varphi(-x)}{1-\Phi(-x)} < \frac{\varphi(-x+\sqrt{q_1})}{1-\Phi(-x+\sqrt{q_1})}\]
and function $\log(f(x))$ is increasing and so does function $f$. Thus
\[\forall \mu\leq 0, \qquad f(\mu)\leq f(0) = \alpha.\]
\end{proof}
We go back to our new test. We can decompose the probability of rejection according to the position of the vector $Y$ in the plane. Denote $Q_1,Q_3$ the nonnegative and non-positive quadrants respectively, and $Q_2, Q_4$ the quadrants $\{y\in\mathbb{R}^2, y_1<0, y_2>0\}$ and $\{y\in\mathbb{R}^2, y_1>0, y_2<0\}$. We have
\begin{multline}
\mathbb{P}_{H_1}\left(LR>q(Y,\alpha)\right) = \mathbb{P}_{H_1}\left(Y_1^2+Y_2^2>q_2| Y\in Q_1\right)\mathbb{P}_{H_1}\left(Y\in Q_1\right) + \\ 
\mathbb{P}_{H_1}\left(Y_1^2>q_1| Y\in Q_4\right)\mathbb{P}_{H_1}\left(Y\in Q_4\right)+ 
\mathbb{P}_{H_1}\left(Y_2^2>q_1| Y\in Q_2\right)\mathbb{P}_{H_1}\left(Y\in Q_2\right)
\label{eqn:All4QuartilesAppendix}
\end{multline}
We treat each one of the conditional probabilities separately. We prove that
\[\mathbb{P}_{H_1}\left(Y_1^2>q_1| Y\in Q_4\right)\leq \alpha.\]
Indeed, since $Y_1$ and $Y_2$ are independent, we have using Lemma \ref{lem:ConditionalIncrease}
\begin{align*}
\mathbb{P}_{H_1}\left(Y_1^2>q_1| Y\in Q_4\right) & = \mathbb{P}_{H_1}\left(Y_1^2>q_1| Y_1\geq 0\right) \\
& \leq \alpha
 \end{align*}
Similarly, we may write the same thing for $Y_2$. We have now
\begin{align}
\mathbb{P}_{H_1}\left(Y_1^2>q_1| Y\in Q_4\right)\leq & \alpha \label{eqn:Q4Appendix} \\
\mathbb{P}_{H_1}\left(Y_2^2>q_1| Y\in Q_2\right)\leq & \alpha.
\label{eqn:Q2Appendix}
\end{align}
We use Lemma \ref{lem:ConditionalIncrease} in order to prove that 
\[\mathbb{P}_{H_1}\left(Y_1^2+Y_2^2>q_2| Y\in Q_1\right)\leq \alpha.\]
We apply the lemma twice. First, let $Z_1=Y_1-\mu_1\sim\mathcal{N}(0,1)$.
\begin{multline*}
\mathbb{P}_{H_1}\left(Y_1^2+Y_2^2>q_2| Y_1\geq 0, Y_2\geq 0\right) = \mathbb{P}_{H_1}\left(Y_1^2>q_2-Y_2^2| Y_1\geq 0, Y_2\geq 0\right)\\
 = \int_0^{q_2}{\frac{\mathbb{P}_{H_1}\left(Y_1^2>q_2-y_2^2, Y_1\geq 0\right)}{\mathbb{P}_{H_1}\left(Y_1\geq 0\right)}\frac{f_{Y_2}(y_2)}{\mathbb{P}_{H_1}\left(Y_2\geq 0\right)}dy_2} + \int_{q_2}^{\infty}{\frac{f_{Y_2}(y_2)}{\mathbb{P}_{H_1}\left(Y_2\geq 0\right)}dy_2} \\
\leq \int_0^{q_2}{\frac{\mathbb{P}_{H_1}\left(Z_1>\sqrt{q_2-y_2^2}\right)}{\mathbb{P}_{H_1}\left(Z_1\geq 0\right)}\frac{f_{Y_2}(y_2)}{\mathbb{P}_{H_1}\left(Y_2\geq 0\right)}dy_2} + \int_{q_2}^{\infty}{\frac{f_{Y_2}(y_2)}{\mathbb{P}_{H_1}\left(Y_2\geq 0\right)}dy_2} \\
\leq \int_0^{q_2}{\frac{\mathbb{P}_{H_1}\left(Z_1^2>q_2-y_2^2, Z_1\geq 0\right)}{\mathbb{P}_{H_1}\left(Z_1\geq 0\right)}\frac{f_{Y_2}(y_2)}{\mathbb{P}_{H_1}\left(Y_2\geq 0\right)dy_2}} + \int_{q_2}^{\infty}{\frac{f_{Y_2}(y_2)}{\mathbb{P}_{H_1}\left(Y_2\geq 0\right)}dy_2} \\
\leq \mathbb{P}_{H_1}\left(Z_1^2+Y_2^2>q_2| Z_1\geq 0, Y_2\geq 0\right) \hspace{7.5cm}
\end{multline*}
We do the same calculation again by swapping the roles of $Y_2$ and $Z_1$. Let $Z_2$ be a standard Gaussian random variable independent of $Z_1$. We have
\[\mathbb{P}_{H_1}\left(Z_1^2+Y_2^2>q_2| Z_1\geq 0, Y_2\geq 0\right)\leq \mathbb{P}_{H_1}\left(Z_1^2+Z_2^2>q_2| Z_1\geq 0, Z_2\geq 0\right)\]
Now, since $Z=(Z_1,Z_2)^T$ is a standard Gaussian random vector $\mathcal{N}(0,I_2)$, then its norm is independent from its direction (\citet[Lemma 3.13.1]{SilvapulleBook}). Therefore,
\begin{align}
\mathbb{P}_{H_1}\left(Y_1^2+Y_2^2>q_2| Y_1\geq 0, Y_2\geq 0\right) & \leq \mathbb{P}_{H_1}\left(Z_1^2+Z_2^2>q_2| Z_1\geq 0, Z_2\geq 0\right) \nonumber\\
& \leq \mathbb{P}_{H_1}\left(Z_1^2+Z_2^2>q_2\right) \nonumber\\
& \leq \alpha. \label{eqn:Q1Appendix}
\end{align}
We conclude by inserting (\ref{eqn:Q1Appendix},\ref{eqn:Q2Appendix},\ref{eqn:Q4Appendix}) in (\ref{eqn:All4QuartilesAppendix}) that
\begin{align*}
\mathbb{P}_{H_1}(LR>q(Y,\alpha)) \leq & \alpha\left(\mathbb{P}_{H_1}\left(Y\in Q_1\right)+\mathbb{P}_{H_1}\left(Y\in Q_2\right)+\mathbb{P}_{H_1}\left(Y\in Q_4\right)\right) \\
\leq & (1-\mathbb{P}_{H_1}(Y\in Q_3))\alpha \\
\leq & \frac{3}{4}\alpha \\
< & \alpha.
\end{align*}


\section{Appendix: An alternative proof for the case of n orthants}\label{append2}
The previous result is easily generalized to $\mathbb{R}^p$. Let $Y\sim\mathcal{N}(\mu,I_p)$. We want to test $H_1:\mu\leq 0$ against all alternatives. The likelihood ratio is given by
\[LR = \|Y-Q\|^2 = \sum_{i=1}^n{Y_i^2\ind{Y_i\geq 0}}.\]
The function $q(Y,\alpha)$ is given by
\[q(Y,\alpha) = q_{\chi^2(N)}, \qquad N = \sum_{i=1}^n{\ind{Y_i\geq 0}}.\]
Let $Q$ be a quadrant with $k$ positive sides. Without loss in generality, assume that
\[Q = \{y\in\mathbb{R}^p, y_1\geq 0, \cdots, y_k\geq 0, y_{k+1}\leq 0,\cdots, y_p\leq 0\}\]
We have
\begin{align*}
\mathbb{P}_{H_1}(LR(Y)>q_k|Y\in Q) = & \mathbb{P}_{H_1}\left(\sum_{i=1}^kY_i^2>q_k\left|Y_1\geq 0, \cdots, Y_k\geq 0, Y_{k+1}\leq 0,\cdots, Y_p\leq 0\right.\right) \\
 = & \mathbb{P}_{H_1}\left(\sum_{i=1}^kY_i^2>q_k\left|Y_1\geq 0, \cdots, Y_k\geq 0\right.\right)
\end{align*}
The second lines comes from the fact that $Y_1,\cdots,Y_n$ are assumed independent. Define the set $S$ as follows
\[S = \left\{(y_2,\cdots,y_k)\in\mathbb{R}^{k-1}| \sum_{i=2}^ky_i^2<q,\;\; y_i\leq 0, \forall i=2,\cdots,k\right\}\]
Denote $y_{2,k}=(y_2,\cdots,y_k)$, $Y_{2,k}=(Y_2,\cdots,Y_k)$ and $\varphi_{2,k}$ the joint density of $Y_{2,k}$. We may now rewrite the conditional probability from above as
\begin{align*}
\mathbb{P}_{H_1}\left(Y_1^2+\|Y_{2,k}\|^2>q_k\left|Y_1\geq 0, Y_{2,k}\geq 0\right.\right) = &  \int_{y_{2,k}\in S}{\frac{\mathbb{P}_{H_1}\left(Y_1>\sqrt{q-\|y_{2,k}\|^2}\right)}{\mathbb{P}_{H_1}(Y_1\geq 0)}\frac{\varphi_{2,k}(y_{2,k})}{\mathbb{P}_{H_1}(Y_{2,k}\geq 0)}dy_1} \\
& + \int_{y_{2,k}\in \mathbb{R}^{k-1}\setminus S}{\frac{\varphi_{2,k}(y_{2,k})}{\mathbb{P}_{H_1}(Y_{2,k}\geq 0)}dy_1}
\end{align*}
Denote $Z_1 = Y_1-\mu_1$. Using Lemma \ref{lem:ConditionalIncrease}, we have
\[\frac{\mathbb{P}_{H_1}\left(Y_1>\sqrt{q-\|y_{2,k}\|^2}\right)}{\mathbb{P}_{H_1}(Y_1\geq 0)}\leq \frac{\mathbb{P}\left(Z_1>\sqrt{q-\|y_{2,k}\|^2}\right)}{\mathbb{P}_{H_1}(Z_1\geq 0)}\]
Thus
\begin{align*}
\mathbb{P}_{H_1}\left(Y_1^2+\|Y_{2,k}\|^2>q_k\left|Y_1\geq 0, Y_{2,k}\geq 0\right.\right) \leq &  \int_{y_{2,k}\in S}{\frac{\mathbb{P}_{H_1}\left(Z_1>\sqrt{q-\|y_{2,k}\|^2}\right)}{\mathbb{P}_{H_1}(Z_1\geq 0)}\frac{\varphi_{2,k}(y_{2,k})}{\mathbb{P}_{H_1}(Y_{2,k}\geq 0)}dy_{2,k}} \\
& + \int_{y_{2,k}\in \mathbb{R}^{k-1}\setminus S}{\frac{\varphi_{2,k}(y_{2,k})}{\mathbb{P}_{H_1}(Y_{2,k}\geq 0)}dy_{2,k}} \\
\leq & \mathbb{P}_{H_1}\left(Z_1^2+\|Y_{2,k}\|^2>q_k\left|Z_1\geq 0, Y_{2,k}\geq 0\right.\right)
\end{align*}
We iterate the previous argument on $Y_2$ and then on $Y_3$ and so on until we get
\[\mathbb{P}_{H_1}\left(\sum_{i=1}^kY_i^2>q_k\left|Y_1\geq 0, \cdots, Y_k\geq 0\right.\right)\leq \mathbb{P}\left(\sum_{i=1}^kZ_i^2>q_k\left|Z_1\geq 0, \cdots, Z_k\geq 0\right.\right)\]
where $Z_i = Y_i - \mu_i\sim\mathcal{N}(0,1)$. Now the joint distribution of the Gaussian vector $Z=(Z_1,\cdots,Z_p)$ is $\mathcal{N}(0,I_p)$, and hence its norm is independent from its direction (\citet[Lemma 3.13.1]{SilvapulleBook}). therefore, we may write
\begin{align*}
\mathbb{P}\left(\sum_{i=1}^kZ_i^2>q_k\left|Z_1\geq 0, \cdots, Z_k\geq 0\right.\right) &= \mathbb{P}\left(\sum_{i=1}^kZ_i^2>q_k\right)\\
& \leq \alpha
\end{align*}
Thus
\begin{align*}
\mathbb{P}_{H_1}\left(LR(Y)>q_k|Y\in Q\right) =& \mathbb{P}_{H_1}\left(\sum_{i=1}^kY_i^2>q_k\left|Y_1\geq 0, \cdots, Y_k\geq 0\right.\right) \\
\leq & \mathbb{P}\left(\sum_{i=1}^kZ_i^2>q_k\right) \\
\leq & \alpha.
\end{align*}
For all the other orthants, an analogue argument holds. Denote now $Q_j = \{y\in\mathbb{R}^p|y_i\geq 0,\forall i\in I_j, y_l\leq 0, \forall l\notin I_j\}$ and $r_j = |I_j|$ for $j=1,\cdots,2^{p}-1$. Let $H_1=Q_{2^p}$. We may now state the following
\begin{align*}
\mathbb{P}_{H_1}\left(LR(Y)>q_k\right) =& \sum_{j=1}^{2^p-1}\mathbb{P}_{H_1}\left(\sum_{i\in I_j}Y_i^2>q_{r_j}|Y\in Q_j\right) \\
= & \sum_{j=1}^{2^p-1}\mathbb{P}_{H_1}\left(\sum_{i\in I_j}Y_i^2>q_{r_j}|Y_i\geq 0, \forall i\in I_j\right) \\
\leq & \left(1-\mathbb{P}(Y\in Q_{2^p})\right)\alpha \\
\leq & \left(1-\frac{1}{2^p}\right)\alpha
\end{align*}
We could now adjust the critical value so that we use the remaining $\frac{1}{2^p}$. Redefine the critical value as
\[\bar{q}(Y,\alpha) = q_{\chi^2(r_j)}\left(1-\frac{2^{p}}{2^{p}-1}\alpha\right), \qquad \text{if } Y\in Q_j.\]
\section{Appendix: A new solution for type A problems which does not require the calculation of the weights: The two dimensional case}\label{append3}
Let $Y\sim(\mu,I_2)$ be a bivariate Gaussian random variable. Consider the testing problem
\[H_0: \mu= 0, \text{ against } H_1: \mu\geq 0.\]
Denote $Q_1,Q_2,Q_3$ and $Q_4$ the four quandrants of the plane. The likelihood ratio is given by
\[LR = \|Y\|^2 - \|Y-\hat{\mu}\|^2\]
where $\hat{\mu}$ is the maximum likelihood estimator of $\mu$ subject to $\mu\geq 0$. It is equal to the projection of $Y$ on the first quadrant $Q_1$. Thus,
\[\hat{\mu} = \left\{\begin{array}{rcl}
(Y_1,Y_2)^T & \text{if} & Y\in Q_1 \\
(0,Y_2)^T & \text{if} & Y\in Q_2 \\
(Y_1,0)^T & \text{if} & Y\in Q_4 \\
(0,0)^T & \text{if} & Y\in Q_3.
\end{array}\right.\]
The LR is then give by
\[LR=\|\Pi(Y|Q_1)\|^2=\|\hat{\mu}\|^2=\left\{\begin{array}{rcl}
Y_1^2+Y_2^2 & \text{if} & Y\in Q_1 \\
Y_2^2 & \text{if} & Y\in Q_2 \\
Y_1^2 & \text{if} & Y\in Q_4 \\
0 & \text{if} & Y\in Q_3.
\end{array}\right.\]
In the literature, testing $H_1$ against $H_2$ at level $\alpha$ using the LRT is done by looking for $c>0$ such that
\[\frac{1}{2}\mathbb{P}(\chi^2(1)>c) + \frac{1}{4}\mathbb{P}(\chi^2(2)>c) \leq \alpha.\]
We propose in this paper to do the test differently. Since the LR is conditionally distributed as a $\chi^2(1)$ when $Y$ has one negative coordinate and one positive one, then we will only test the LR against the quantile of the $\chi^2(1)$, and there is no need to consider the whole mixture of chi-squares. When $Y$ is in the positive quadrant, the LR is conditionally distributed as a $\chi^2(2)$, and we will test the LR against the quantile of the $\chi^2(2)$. In other words, we define the random quantile
\[ q(Y,\alpha) = 
 \left\{\begin{array}{rcl}
0 & \text{if} & Y\in Q_3 \\
q_1 & \text{if} & Y\in Q_2, \text{ or }, Y\in Q_4\\
q_2 & \text{if} & Y\in Q_1.
\end{array}\right.
\]
We test $H_0$ against $H_1$ using the rejection region 
\[\left\{LR>q(Y,\alpha)\right\}.\]
Our claim is that 
\[\mathbb{P}_{H_0}\left(LR>q(Y,\alpha)\right)\leq \alpha.\]
Indeed, by conditioning on the position of $Y$ in the plane, we may write
\begin{multline*}
\mathbb{P}_{H_0}\left(LR>q(Y,\alpha)\right) = \mathbb{P}_{H_0}\left(Y_1^2+Y_2^2>q_2| Y\in Q_1\right)\mathbb{P}_{H_0}\left(Y\in Q_1\right) + \\ 
\mathbb{P}_{H_0}\left(Y_1^2>q_1| Y\in Q_4\right)\mathbb{P}_{H_0}\left(Y\in Q_4\right)+ 
\mathbb{P}_{H_0}\left(Y_2^2>q_1| Y\in Q_2\right)\mathbb{P}_{H_0}\left(Y\in Q_2\right)
\end{multline*}
Since under the null $Y$ is distributed as $\mathcal{N}(0,I_2)$, its norm and its direction are independent. Therefore
\begin{align*}
\mathbb{P}_{H_0}\left(LR>q(Y,\alpha)\right) = & \mathbb{P}_{H_0}\left(Y_1^2+Y_2^2>q_2\right)\mathbb{P}_{H_0}\left(Y\in Q_1\right) + \mathbb{P}_{H_0}\left(Y_1^2>q_1\right)\mathbb{P}_{H_0}\left(Y\in Q_4\right)\\ 
& + \mathbb{P}_{H_0}\left(Y_2^2>q_1\right)\mathbb{P}_{H_0}\left(Y\in Q_2\right) \\
\leq & \frac{3}{4}\alpha
\end{align*}
Therefore, we do not need to use a mixture of chi-squares in order to test $H_0$ against $H_1$, and we can instead choose which chi-square to use according to the position of $Y$ in the plane. Moreover, we could also make the test even a little more powerful by adjusting it to provide exactly a level-$\alpha$ test, that is
\[ q(Y,\alpha) = 
 \left\{\begin{array}{rcl}
0 & \text{if} & Y\in Q_3 \\
q_{\chi^2(1)}\left(1-\frac{4}{3}\alpha\right) & \text{if} & Y\in Q_2, \text{ or }, Y\in Q_4\\
q_{\chi^2(2)}\left(1-\frac{4}{3}\alpha\right) & \text{if} & Y\in Q_1.
\end{array}\right.
\]

\section{Appendix: A new solution for type A problems which does not require the calculation of the weights: The general case}\label{append4}
Let $Y\sim\mathcal{N}(\mu,V)$. We test $H_0:\mu=0$ against $H_1:\mu\in\C$ with $\C$ a polyhedral cone in $\mathbb{R}^p$. The LR is given by
\[LR = \bar{\chi}^2(V,\C) = Y^TV^{-1}Y - \min_{\mu\in\C} (Y-\mu)^TV^{-1}(Y-\mu)\]
According to \cite[Proposition 3.4.1]{SilvapulleBook}, we have
\[LR = \|\Pi(Y|\C)\|_V^2.\]
Using Lemma 3.13.5 from \cite{SilvapulleBook}, the projection on a polyhedral cone is characterized through the projection on the linear spaces spanned by its faces. In other words, there exists a collection of faces of $\C$, say $\{F_1,\cdots,F_K\}$ such that the collection of their relative interiors, $\{ri(F_1),\cdots,ri(F_K)\}$, forms a partition of $\C$. Further,
\[\|\Pi(Y|\C)\|_V^2 = \sum_{i=1}^K{\ind{\Pi(Y|\C)\in ri(F_i)}\|P_iY\|_V^2}\]
where $P_i$ is the projection matrix onto the linear space spanned by $F_i$. Denote $r_i=rank(P_i)$, and define the random critical value
\begin{equation}
q(\C,Y,\alpha) = \sum_{i=1}^K{\ind{\Pi\left(Y|\C\right)\in ri(F_i)}q_{\chi^2(r_i)}(1-\alpha)}
\label{eqn:SelectiveQuantileTypeA}
\end{equation}
Our claim is that 
\[
\mathbb{P}_{H}\left(LR(\C)>q(\C,Y,\alpha)\right) \leq \alpha.
\]
We start with the case when $V=I_p$. We have
\begin{align*}
\mathbb{P}_{H}\left(LR(\C)>q(\C,Y,\alpha)\right) = & \sum_{j=1}^K{\mathbb{P}_{H_0}\left(\|P_jY\|^2>q_{r_j}|\Pi(Y|\C)\in ri(F_j)\right)\mathbb{P}\left(\Pi(Y|\C)\in ri(F_j)\right)}
\end{align*}
Lemma 3.13.2 from \cite{SilvapulleBook} states that the event
\[\left\{\Pi\left(Y|\C\right)\in ri(F_j)\right\}\]
is equivalent to the event
\[\{P_jY\in ri(F_j),\; (I_n-P_j)Y\in F_j^{\bot}\cap \C\}\]
and since $P_j$ is a projection matrix, the random variables $P_jY$ and $(I_n-P_j)Y$ are independent. Thus
\begin{align*}
\mathbb{P}_{H_0}\left(\|P_jY\|^2>q_{r_j}|\Pi(Y|\C)\in ri(F_j)\right) = & \mathbb{P}_{H_0}\left(\|P_jY\|^2>q_{r_j}|P_jY\in ri(F_j),\; (I_n-P_j)Y\in F_j^{\bot}\cap \C\right) \\
= & \mathbb{P}_{H_0}\left(\|P_jY\|^2>q_{r_j}|P_jY\in ri(F_j)\right)
\end{align*}
Finally, since $Y\sim\mathcal{N}(0,I_n)$, then according to Lemma 3.13.3 from \cite{SilvapulleBook}, the distribution of $\|P_i X\|^2$ conditionally on $P_iY\in ri(F_j)$ is $\chi^2(r_j)$. Thus
\[\mathbb{P}_{H_0}\left(\|P_jY\|^2>q_{r_j}|P_jY\in ri(F_j)\right)\leq \alpha\]
Hence
\[\mathbb{P}_{H}\left(LR(\C)>q(\C,Y,\alpha)\right) \leq  (1-\mathbb{P}_{H_0}(\Pi(Y|\C)\in ri(F_0)))\alpha \]
where $F_0=\{0\}$. Besides, the event $\{\Pi(Y|\C)\in ri(F_0)\}$ is equivalent to the event $\{Y\in \C^{o}\}$. In case we have a way to calculate the probability $\mathbb{P}_{H_0}(Y\in\C^{o})$, then we could recalibrate the critical value and gain more power. For example, in the case of orthants, $\mathbb{P}_{H_0}(Y\in\C^{o})=2^{-p}$.\\
The end of the proof is the same as for Proposition \ref{prop:typeBGeneral}.

\bibliographystyle{plainnat}
\bibliography{biblioFile}
\end{document}